\pdfoutput=1

\documentclass[11pt]{article}

\usepackage{ifthen}

\newcommand{\Comment}[1]{\ignorespaces}

\usepackage{algorithmic}
\usepackage{algorithm}

\usepackage{amsmath, amssymb}

\usepackage{amsthm}

\usepackage{ifpdf}

\usepackage{caption}
\usepackage{subcaption}

\ifpdf
\usepackage[pdftex]{graphicx}
\usepackage[pdftex,%
  ocgcolorlinks,%
  linkcolor=blue,%
  filecolor=blue,%
  citecolor=blue,%
  urlcolor=blue]{hyperref}
\else
\usepackage[dvips]{graphicx}
\fi

\ifthenelse{\isundefined{\NoGenerateLetterSize}}
{
\ifpdf
\setlength{\pdfpagewidth}{8.5in}
\setlength{\pdfpageheight}{11in}
\else

\fi
}
{}

\newtheorem{theorem}{Theorem}[section]
\newtheorem{lemma}[theorem]{Lemma}
\newtheorem{observation}[theorem]{Observation}
\newtheorem{corollary}[theorem]{Corollary}

\theoremstyle{definition}

\newtheorem*{definition*}{Definition}
\theoremstyle{plain}

\newenvironment{DenseItemize}[0]
{\begin{itemize} \itemsep0pt \parskip0pt \parsep0pt}
{\end{itemize}}

\newenvironment{DenseEnumerate}[0]
{\begin{enumerate} \itemsep0pt \parskip0pt \parsep0pt}
{\end{enumerate}}

\newcommand{\Id}{\mathrm{id}}
\newcommand{\Image}{\mathrm{Im}}
\newcommand{\Domain}{\mathrm{Dom}}
\newcommand{\Cost}{\mathit{c}}
\newcommand{\Ben}{\mathit{b}}
\newcommand{\Null}{\mathrm{NULL}}
\newcommand{\Opt}{\ensuremath{\mathtt{OPT}}}
\newcommand{\Alg}{\ensuremath{\mathtt{ALG}}}
\newcommand{\Rationals}{\mathbb{Q}}

\newcommand{\EdgeId}{\mathrm{eid}}
\newcommand{\MinEdgeCost}{\mathrm{emin}}
\newcommand{\Eff}{\mathrm{eff}}
\newcommand{\Lev}{\mathrm{lev}}
\newcommand{\Expectation}{\mathbb{E}}
\newcommand{\Probability}{\mathbb{P}}
\newcommand{\Entropy}{\mathit{H}}
\newcommand{\Dummy}{\mathrm{d}}
\newcommand{\SSSC}[0]{\ensuremath{\mathtt{SSSC}}}
\newcommand{\COVER}[0]{\ensuremath{\mathtt{COVER}}}

\newcommand{\Sect}{Sec.}
\newcommand{\Thm}{Thm.}
\newcommand{\Lem}{Lem.}
\newcommand{\Cor}{Cor.}
\newcommand{\Obs}{Obs.}

\begin{document}

\title{Semi-Streaming Set Cover
\\
(Full Version)%
}

\author{
Yuval Emek\thanks{
Technion, Israel.
Email: \texttt{yemek@ie.technion.ac.il}.
}
\and
Adi Ros\'{e}n\thanks{
CNRS and Universit\'{e} Paris Diderot, France.
Email: \texttt{adiro@liafa.univ-paris-diderot.fr}.
Research supported in part by ANR project RDAM.
}
}

\date{}

\begin{titlepage}

\maketitle

\begin{abstract}
This paper studies the set cover problem under the semi-streaming model.
The underlying set system is formalized in terms of a hypergraph $G = (V, E)$
whose edges arrive one-by-one and the goal is to construct an edge cover $F
\subseteq E$ with the objective of minimizing the cardinality (or cost in
the weighted case) of $F$.
We consider a parameterized relaxation of this problem, where given some
$0 \leq \epsilon < 1$,
the goal is to construct an edge $(1 - \epsilon)$-cover, namely, a
subset of edges incident to all but an $\epsilon$-fraction of the vertices (or
their benefit in the weighted case).
The key limitation imposed on the algorithm is that its space is limited to
(poly)logarithmically many bits per vertex.

Our main result is an asymptotically tight trade-off between $\epsilon$ and
the approximation ratio:
We design a semi-streaming algorithm that on input graph $G$, constructs a
succinct data structure $\mathcal{D}$ such that for every
$0 \leq \epsilon < 1$,
an edge $(1 - \epsilon)$-cover that approximates the optimal edge
\mbox{($1$-)cover} within a factor of
$f(\epsilon, n)$
can be extracted from $\mathcal{D}$ (efficiently and with no additional space
requirements),
where
\[
f(\epsilon, n) =
\left\{
\begin{array}{ll}
O (1 / \epsilon), & \text{if } \epsilon > 1 / \sqrt{n} \\
O (\sqrt{n}), & \text{otherwise}
\end{array}
\right. \, .
\]
In particular for the traditional set cover problem we obtain an
$O(\sqrt{n})$-approximation.
This algorithm is proved to be best possible by establishing a family
(parameterized by $\epsilon$) of matching lower bounds.
\end{abstract}

\renewcommand{\thepage}{}
\end{titlepage}

\pagenumbering{arabic}

\section{Introduction}
Given a \emph{set system} consisting of a \emph{universe} of items and a
collection of item sets, the goal in the \emph{set cover} problem is to
construct a minimum cardinality subcollection of sets that covers the whole
universe.
This problem is fundamental to combinatorial optimization with applications
ranging across many different domains.
It is one of the 21 problems whose NP-hardness was established by Karp in
\cite{Karp1972} and its study has led to the development of various
techniques in the field of approximation algorithms (see, e.g.,
\cite{Vazirani2001}).

In this paper, we investigate the set cover problem under the
\emph{semi-streaming} model \cite{Feigenbaum2005}, where the sets arrive
one-by-one and the algorithm's space is constrained to maintaining a small
number of bits per item (cf.\ the \emph{set-streaming} model of
\cite{SahaG2009}).
In particular, we are interested in the following two research questions:
(1) What is the best approximation ratio for the set cover problem under such
memory constraints?
(2) How does the answer to (1) change if we relax the set cover notion so
that the set subcollection is required to cover only a $\delta$-fraction of
the universe?

On top of the theoretical interest in the aforementioned research
questions, studying the set cover problem under the semi-streaming model is
justified by several practical applications too.
For example, Saha and Getoor~\cite{SahaG2009} describe the setting of a web
crawler that iterates a large collection of blogs, listing the topics covered
by each one of them.
A user interested in a certain set of topics can run a semi-streaming set
cover algorithm with relatively small memory requirements to identify a
subcollection of blogs that covers her desired topics.

\paragraph{The model.}
In order to fit our terminology to the graph theoretic terminology
traditionally used in the semi-streaming literature (and also to ease up the
presentation), we use an equivalent formulation for the set cover problem in
terms of edge covers in hypergraphs:
Consider some \emph{hypergraph} $G = (V, E)$, where $V$ is a set of $n$
\emph{vertices} and $E$ is a (multi-)set of $m$ \emph{hyperedges} (henceforth
\emph{edges}), where each edge $e \in E$ is an arbitrary non-empty subset $e
\subseteq V$.
Assume hereafter that $G$ does not admit any isolated vertices, namely, every
vertex is incident to at least one edge.
We say that an edge subset $F \subseteq E$ \emph{covers} $G$ if every
vertex in $V$ is incident to some edge in $F$.
The goal of the \emph{edge cover} problem is to construct a subset $F
\subseteq E$ of edges that covers $G$, where the objective is to minimize the
cardinality $|F|$.

A natural relaxation of the covering notion asks to cover some fraction of the
vertices in $V$:
Given some $0 < \delta \leq 1$, we say that an edge subset $F
\subseteq E$ \emph{$\delta$-covers} $G$ if at least $\delta n$ vertices are
incident to the edges in $F$, namely,
$\left| V(F) \right| \geq \delta n$,
where
$V(F) = \{ v \in V \mid \exists e \in F \text{ s.t. } v \in e \}$.
Under this terminology, a cover of $G$ is referred to as a $1$-cover.
This raises a bi-criteria optimization version of the set cover problem, where
the goal is to construct an edge subset $F \subseteq E$ that
$\delta$-covers $G$ with the objective of minimizing $|F|$ and maximizing
$\delta$.
In this paper, we focus on approximation algorithms, where the cardinality of
$F$ is compared to that of an optimal edge \mbox{($1$-)cover} of $G$.

In the \emph{weighted} version of the edge cover problem, the hypergraph $G$
is augmented with vertex \emph{benefits}
$\Ben : V \rightarrow \Rationals_{> 0}$
and edge \emph{costs}
$\Cost : E \rightarrow \Rationals_{> 0}$.
The edge cover definition is generalized so that edge subset $F
\subseteq E$ is said to \emph{$\delta$-cover} $G$ if the benefit of the
vertices incident to the edges in $F$ is at least a $\delta$-fraction of the
total benefit, namely,
$\Ben(V(F)) \geq \delta \cdot \Ben(V)$,
where
$\Ben(U) = \sum_{v \in U} \Ben(v)$
for every vertex subset $U \subseteq V$.
The goal is then to construct an edge subset $F$ that $\delta$-covers
$G = (V, E, \Ben, \Cost)$,
where the objective is to maximize $\delta$ and minimize the cost of $F$,
denoted
$\Cost(F) = \sum_{e \in F} \Cost(e)$.

Under the \emph{semi-streaming} model, the execution is partitioned into
discrete time steps and the edges in $E$ are presented one-by-one so that edge
$e_{t} \in E$ is presented at time $t = 0, 1, \dots, m - 1$, listing all
vertices $v \in e_{t}$;\footnote{
With the exception of our related work discussion, all semi-streaming
algorithms in this paper make a single (one way) pass over the input
hypergraph.
}
in the weighted version, the cost of $e_{t}$ and the benefits of the vertices
it contains are also listed.
The key limitation imposed on the algorithm is that its space is limited;
specifically, we allow the algorithm to maintain $\log^{O (1)} |G|$
bits per vertex, where $|G|$ denotes the number of bits in the standard binary
encoding of $G$.
Each edge $e \in E$ is associated with a unique \emph{identifier} $\Id(e)$ of
size $O (\log m)$ bits, say, the time $t$ at which edge $e_{t}$ is presented.
We may sometimes use the identifier $\Id(e)$ when we actually refer to the
edge $e$ itself, e.g., replacing $\Cost(e)$ with $\Cost(\Id(e))$;
our intention will be clear from the context.

In contrast to the random access memory model of computation, where given a
collection $\mathcal{I}$ of identifiers, one can easily determine which vertex
in $V$ is incident to which of the edges whose identifiers are in
$\mathcal{I}$ simply by examining the input, under the semi-streaming model,
the collection $\mathcal{I}$ by itself typically fails to provide this
information.
Therefore, instead of merely returning the identifiers of some edge
$\delta$-cover, we require that the algorithm outputs a
\emph{$\delta$-cover certificate} $\chi$ for $G$ which is a partial function
from $V$ to $\{ \Id(e) \mid e \in E \}$ with
\emph{domain}
\[
\Domain(\chi) = \{ v \in V \mid \chi \text{ is defined over } v \}
\]
and \emph{image}
\[
\Image(\chi) = \{ \Id(e) \mid \exists v \in \Domain(\chi) \text{ s.t. }
\chi(v) = \Id(e) \}
\]
that satisfies
(1) if $v \in \Domain(\chi)$ and $\chi(v) = \Id(e)$, then $v \in e$; and
(2) $\Ben(\Domain(\chi)) \geq \delta \cdot \Ben(V)$.
By definition, the image of $\chi$ consists of the identifiers of the edges in
some edge $\delta$-cover $F$ of $G$ and the quality of the $\delta$-cover
certificate $\chi$ is thus measured in terms of
$\Cost(\Image(\chi)) = \Cost(F)$.

\paragraph{Our contribution.}
Consider some unweighted hypergraph $G = (V, E)$ with optimal edge $1$-cover
$\Opt$.
We design a deterministic semi-streaming algorithm, referred to as \SSSC{}
(acronym of the paper's title), for the edge ($\delta$-)cover problem that
given some $0 \leq \epsilon < 1$,
outputs a $(1 - \epsilon)$-cover certificate $\chi_{\epsilon}$ for $G$ with
image of cardinality
$|\Image(\chi_{\epsilon})| = O (\min \{ 1 / \epsilon, \sqrt{n} \} \cdot
|\Opt|)$.\footnote{
Define $\min \{ 1 / x, y \} = y$ when $x = 0$.
}
This result is extended to the weighted case, where $G = (V, E, \Ben, \Cost)$,
showing that
$\Cost(\Image(\chi_{\epsilon})) = O (\min \{ 1 / \epsilon, \sqrt{n} \} \cdot
\Cost(\Opt))$
(see \Thm{}\ \ref{theorem:upper-bound} and
\ref{theorem:upper-bound-aspect-ratio}).
In particular, for the edge \mbox{($1$-)cover} problem, we obtain an $O
(\sqrt{n})$-approximation for both the weighted and unweighted cases.

On the negative side, we prove that for every
$\epsilon \geq 1 / \sqrt{n}$,
if a randomized semi-streaming algorithm for the set cover problem
outputs a $(1 - \epsilon)$-cover certificate $\chi$ for $G$, then it cannot
guarantee that
$\Expectation[|\Image(\chi)|] = o (|\Opt| / \epsilon)$
(see \Thm{}~\ref{theorem:lower-bound}).
This demonstrates that the approximation guarantee of our algorithm is
asymptotically optimal for the whole range of parameter
$0 \leq \epsilon < 1$
even for randomized algorithms.

Notice that \SSSC{} has the attractive feature that the (near-linear
size) data structure $\mathcal{D}$ it maintains is oblivious to
the parameter $\epsilon$.
That is, the algorithm processes the stream of edges with no knowledge of
$\epsilon$,  generating the data structure $\mathcal{D}$, and the promised $(1
- \epsilon)$-cover certificate $\chi_{\epsilon}$ can be efficiently extracted
from $\mathcal{D}$ (with no additional space requirements) for every
$0 \leq \epsilon < 1$
(in fact several such covers for different values of $\epsilon$ can be
extracted).
From a bi-criteria optimization perspective, our lower bound implies that the
parameterized collection
$\{ \chi_{\epsilon} \}_{0 \leq \epsilon < 1}$
encoded in $\mathcal{D}$ is an (asymptotically) optimal solution frontier
(cf. Pareto optimality).

Using a simple adjustment of the randomized rounding technique for set cover
(see, e.g., \cite{Vazirani2001}), it is not difficult to show that a basic
feasible solution to the linear program relaxation $\mathcal{P}$ of a given
set cover instance also serves as a compact data structure from which a $(1 -
\epsilon)$-cover certificate $\chi_{\epsilon}$ can be extracted for every
$0 \leq \epsilon < 1$.
In fact, the approximation ratio obtained this way is  better
than ours, namely, $O (\log (1 / \epsilon))$.
However, our lower bound shows that this approach cannot be applied --- and in
passing, that $\mathcal{P}$ cannot be solved --- under the semi-streaming
model.

Can our tight lower bound be an artifact of the requirement that the
algorithm outputs a cover \emph{certificate}?
We nearly eliminate this possibility by proving that for every constant
$c > 0$
and for every
$\epsilon \geq n^{-1 / 2 + c}$,
even if the randomized algorithm only guarantees an ``uncertified'' output,
i.e., only the identifiers of the edges in some edge $(1 - \epsilon)$-cover
$F$ of $G$ are returned, then the cardinality of $F$ must still be large,
specifically,
$|F| = \Omega \left( \frac{\log\log n}{\log n} \cdot |\Opt|  / \epsilon
\right)$,
where $\Opt$ in this case is proportional to $\epsilon^{2} n$
(see \Thm{}~\ref{theorem:lower-bound-uncertified}).\footnote{
By using a reduction from the index function studied in communication
complexity \cite{KremerNR1999}, one can show that there does not exist a
semi-streaming algorithm that distinguishes between hypergraphs admitting a
constant size edge cover and hypergraphs that cannot be covered by less than
$n^{\alpha}$ edges for any constant $0 < \alpha < 1 / 2$.
This lower bound is more attractive in the sense that it applies already to
the decision version of the set cover problem however, to the best of our
understanding, in contrast to the constructions of the present paper, this
result cannot be generalized to $(1 - \epsilon)$-covers for values of
$\epsilon \gg 1 / \sqrt{n}$.
}

\paragraph{Related work.}
The work most closely related to the present paper is probably the one
presented in Saha and Getoor's paper \cite{SahaG2009} that also considers the
set cover problem under the semi-streaming model (referred to as
\emph{set-streaming} in \cite{SahaG2009}) formulated as the edge cover problem
in hypergraphs.
Saha and Getoor design a $4$-approximation semi-streaming algorithm for the
\emph{maximum coverage} problem that given a hypergraph $G = (V, E)$ and a
parameter $k$, looks for $k$ edges that cover as many vertices as possible.
Based on that, they observe that an $O (\log n)$-approximation for the optimal
set cover can be obtained in $O (\log n)$ passes over the
input (this can be achieved based on our semi-streaming algorithm as well).
Using the terminology of the present paper, Saha and Getoor's maximum coverage
algorithm is very efficient for obtaining edge $(1 - \epsilon)$-covers as long
as $\epsilon$ is large, but it does not provide any (single pass) guarantees
for $\epsilon < 3 / 4$.
In contrast, our algorithm has asymptotically optimal (single pass) guarantees
for any $0 \leq \epsilon < 1$.
Another paper that considers semi-streaming algorithms in hypergraphs is that
of Halld\'{o}rsson et al.~\cite{HalldorssonHLS10} that studies the independent
set problem.

The semi-streaming model was introduced by Feigenbaum et
al.~\cite{Feigenbaum2005} for graph theoretic problems, where the edges of an
$n$ vertex input graph arrive sequentially and the algorithm is allowed to
maintain only
$\log^{O (1)} n$
bits of memory per vertex.
Since the number of bits required to encode an $n$ vertex graph is
$n^{O (1)}$, the space-per-vertex bound used in the present paper can be
viewed as a generalization of that of Feigenbaum et al.\ from graphs to
hypergraphs.
In any case, concerns regarding the comparison between the space bound used in
the present paper and that of \cite{Feigenbaum2005} can be lifted by
restricting attention to hypergraphs with
$m \leq 2^{\log^{O (1)} n}$
edges
(refer to \Sect{}~\ref{section:Algorithm} for a further discussion of
the space bounds of our algorithm).

Various graph theoretic problems have been treated under the semi-streaming
model.
These include
matching \cite{McGregor05,EpsteinLMS10,kmm12},
diameter and shortest path \cite{Feigenbaum2005,FKMSZ08},
min-cut and sparsification \cite{AG09,KelnerL2013},
graph spanners \cite{FKMSZ08}, and
independent set \cite{HalldorssonHLS10,EmekHR2012}.

Several variants of the set cover problem, all different than the problem
studied in the present paper, have been investigated under the model of online
computation.
Alon et al.~\cite{AlonAABN2009} focus on the online problem in which some
master set system is known in advance and an unknown subset of its items
arrive online;
the goal is to cover the arriving items, minimizing the number of sets used
for that purpose.
Another online variant of the set cover problem is studied by Fraigniaud et
al.~\cite{FraigniaudHPRR2013}, where the sets arrive online, but not all items
have to be covered.
Here, each item is associated with a penalty and the cost of the algorithm is
the sum of the total cost of the sets chosen for the partial cover and the
total penalty of the uncovered items.

Note that under the online computation model, there is a trivial linear lower
bound for the problem studied in the present paper if preemption is not
allowed.
If preemption is allowed, then the problem becomes interesting only under a
slightly stronger definition for the competitive ratio:
The performance of the algorithm is measured via the maximum over time $t$ of
the ratio
$\Alg_{t} / \Opt_{t}$,
where $\Opt_{t}$ is the cost of an optimal set cover for the set system
presented up to time $t$ and $\Alg_{t}$ is the cost of the set cover
maintained by the algorithm for that set system.
The set cover algorithm presented in the present paper is, in
fact, also an online algorithm for this problem with competitive ratio $O
(\sqrt{n})$.
The lower bound(s) established in the present paper can be slightly modified to
show that this is optimal.

Closely related to our notion of cover certificate is the \emph{universal set
cover} problem \cite{JiaLNRS2005,GrandoniGLMSS2013}, where given a set system,
the goal is to construct a mapping $f$ from the items to the sets containing
them so that for every item subset $X$, the cost of the image of $X$ under $f$
is as close as possible to the cost of a minimum set cover for $X$.
This problem resembles our guarantee that the promised $(1 - \epsilon)$-cover
certificate can be extracted from the data structure for every $\epsilon$
however, it is much stronger in the sense that it guarantees a small cover for
every item subset, rather than the existence of a ``good'' item subset for
every $\epsilon$.
To the best of our knowledge, the universal set cover problem has not been
studied under the semi-streaming model.

\paragraph{Techniques' overview.}
The main procedure of our algorithm \SSSC{} (referred to as \COVER{})
maintains for each vertex $v \in V$, a variable $\Eff(v)$.
This variable captures the ratio of the benefit of the last
\emph{effective} subset $T \subseteq e_{t}$ that covered $v$ to the cost of
$e_{t}$, where subset $T \subseteq e_{t}$ is said to be effective if
$\Ben(T) / \Cost(e_{t}) \geq 2 \cdot \Eff(u)$
for every $u \in T$.
This means, in particular, that the variable $\Eff(v)$ doubles with every
update.
(Note that \COVER{} actually maintains the logarithm of this $\Eff(v)$
variable for each vertex $v$, but the main idea is the same.)
By picking the effective subset $T \subseteq e_{t}$ that
maximizes $\Ben(T)$, we ensure that the collection of vertices $v \in V$
admitting high values of $\Eff(v)$ satisfies some desirable properties.
Specifically, a careful analysis shows that upon termination of the input
stream, there exists some threshold $\rho$ such that the total benefit of
vertices $v \in V$ with
$\Eff(v) \leq \rho$
is at most
$\epsilon \cdot \Ben(V)$,
whereas the total cost of the edges corresponding to the effective subsets of
the vertices $v \in V$ with $\Eff(v) > \rho$ is
$O (\Cost(\Opt) / \epsilon)$.
Invoking procedure \COVER{} on a hypergraph with the same edge costs and
uniform vertex benefits (in parallel to the invocation of \COVER{} on the
original input hypergraph) enables us to produce an edge $1$-cover that $O
(\sqrt{n})$-approximates $\Cost(\Opt)$.

The bad hypergraphs that lie at the heart of our lower bound are
constructed based on an \emph{affine plane} $\mathcal{A} = (P, L)$ with
$q^{2}$ points and $q (q + 1)$ lines (see, e.g., \cite{LindnerR2011}) by
randomly partitioning each line in $L$ into two edges (more edges in the
``uncertified'' version of the lower bound).
After presenting the two edges corresponding to all lines in $L$, we present
one additional edge $e^{*}$ that contains the points of all but
$r \approx \epsilon q$
random lines from some random angle $A_{i}$ of $\mathcal{A}$.
An optimal edge cover consists of the edge $e^{*}$ and the $2 r = O (\epsilon
q)$ edges corresponding to the $r$ lines missing from $e^{*}$.
Using careful information theoretic arguments, we show that any low space
deterministic algorithm must use many lines from angles other than
$A_{i}$ to construct a $(1 - \epsilon)$-cover $F$.
The properties of affine planes guarantee that the expected cardinality of $F$
is $\Omega (q)$.
By Yao's principal, our lower bound is translated from deterministic
algorithms to randomized ones.

\section{A semi-streaming algorithm}
\label{section:Algorithm}
Our goal in this section is to design a semi-streaming algorithm for the edge
($\delta$-)cover problem in hypergraphs.
The algorithm, referred to as \SSSC{}, is presented in
\Sect{}~\ref{section:algorithm-description} and its approximation ratio is
analyzed in \Sect{}~\ref{section:algorithm-analysis}.
For the sake of simplicity, we first assume that
all numerical values (vertex benefits and edge costs) are encoded using $O
(\log n)$ bits.
Under this assumption, the space bounds of \SSSC{} are
quite trivial and the analysis in \Sect{}~\ref{section:algorithm-analysis}
yields Theorem~\ref{theorem:upper-bound-assumption}.

\begin{theorem} \label{theorem:upper-bound-assumption}
On a weighted input hypergraph $G = (V, E, \Ben, \Cost)$ with numerical values
encoded using $O (\log n)$ bits,
our algorithm uses
$O (n \log (n + m))$
space, processes each input edge $e_{t} \in E$ in
$O (|e_{t}| \log |e_{t}|)$
time, and produces a data structure $\mathcal{D}$ with the following
guarantee:
For every
$0 \leq \epsilon < 1$,
a $(1 - \epsilon)$-cover certificate $\chi_{\epsilon}$ for $G$ such that
\[
\Cost(\Image(\chi_{\epsilon}))
=
O \left( \min \left\{ 1 / \epsilon, \sqrt{n} \right\} \cdot
\Cost(\Opt) \right)
\]
can be extracted from $\mathcal{D}$ in time $O (n \log n)$ with no additional
space requirements, where $\Opt$ stands for an optimal edge \mbox{($1$-)cover}
of $G$.
\end{theorem}

\Sect{}~\ref{section:algorithm-space} is dedicated to lifting the assumption
on the numerical values.
The following definitions are necessary for the discussion of the results we
obtain without this assumption:
\[
\Ben^{\lg} = \lg \left\lceil \max_{v \in V} \left\{ \Ben(v), \Ben(v)^{-1} \right\}
\right\rceil \,
\quad
\Cost^{\lg} = \lg \left\lceil \max_{e \in E} \left\{ \Cost(e), \Cost(e)^{-1}
\right\} \right\rceil \,
\quad
\Cost^{\Delta} = \lg \left\lceil \frac{\max_{e \in E} \Cost(e)}{\min_{e \in E}
\Cost(e)} \right\rceil \, ,
\]
where the last parameter captures the number of bits required to encode the
edge costs \emph{aspect ratio}.\footnote{
Throughout, $\lg$ denotes logarithm to the base of $2$.
}
Note that the encoding size $|G|$ of the input weighted hypergraph $G = (V, E,
\Ben, \Cost)$ is at least
$\Ben^{\lg} + \Cost^{\lg}$.
Moreover,
$\Cost^{\Delta}$
is always at most
$2 \Cost^{\lg}$,
but it may be much smaller than that.

Our results are cast in \Thm{}\ \ref{theorem:upper-bound} and
\ref{theorem:upper-bound-aspect-ratio}, where the former generalizes
\Thm{}~\ref{theorem:upper-bound-assumption} and the latter has a better space
bound, but slightly worse run-time guarantee.
Another drawback of \Thm{}~\ref{theorem:upper-bound-aspect-ratio} is
that it requires that the parameters $n$ and $\epsilon$ are known to the
algorithm in advance in contrast to \Thm{}
\ref{theorem:upper-bound} and \ref{theorem:upper-bound-assumption}
that do not require an apriori knowledge of any global parameter.

\sloppy
\begin{theorem} \label{theorem:upper-bound}
On a weighted input hypergraph $G = (V, E, \Ben, \Cost)$,
our algorithm uses
$O \left( n \log \left( n + m + \Ben^{\lg} + \Cost^{\lg} \right) \right)$
space, processes each input edge $e_{t} \in E$ in
$O (|e_{t}| \log |e_{t}|)$
time, and produces a data structure $\mathcal{D}$ with the following
guarantee:
For every
$0 \leq \epsilon < 1$,
a $(1 - \epsilon)$-cover certificate $\chi_{\epsilon}$ for $G$ such that
\[
\Cost(\Image(\chi_{\epsilon}))
=
O \left( \min \left\{ 1 / \epsilon, \sqrt{n} \right\} \cdot
\Cost(\Opt) \right)
\]
can be extracted from $\mathcal{D}$ in time $O (n \log n)$ with no additional
space requirements, where $\Opt$ stands for an optimal edge \mbox{($1$-)cover}
of $G$.
\end{theorem}
\par\fussy

\begin{theorem} \label{theorem:upper-bound-aspect-ratio}
On a weighted input hypergraph $G = (V, E, \Ben, \Cost)$,
for any $0 \leq \epsilon < 1$,
our algorithm uses
$O \left( \log \left( \Ben^{\lg} + \Cost^{\lg} \right) + n \log \left( n + m +
\Cost^{\Delta} \right) \right)$
space, processes each input edge $e_{t} \in E$ in
$O (n \log n)$
time, and outputs a $(1 - \epsilon)$-cover certificate $\chi_{\epsilon}$
for $G$ such that
\[
\Cost(\Image(\chi_{\epsilon}))
=
O \left( \min \left\{ 1 / \epsilon, \sqrt{n} \right\} \cdot
\Cost(\Opt) \right) \, ,
\]
where $\Opt$ stands for an optimal edge \mbox{($1$-)cover} of $G$.
\end{theorem}

\subsection{The Algorithm}
\label{section:algorithm-description}
In what follows we consider some weighted hypergraph $G = (V, E, \Ben, \Cost)$
with optimal edge \mbox{($1$-)cover} $\Opt$.
The main building block of algorithm \SSSC{} is a procedure referred to as
\COVER{}.
This procedure processes the stream of edges and outputs for every node $v \in
V$, an identifier of an edge $e$ that covers it, together with an integer
variable that intuitively captures the quality of edge $e$ in covering $v$.
Algorithm \SSSC{} uses two parallel invocations of \COVER{}, one on the input
graph $G$ and one on some modification of $G$, and upon termination of the
input stream, extracts the desired cover certificate from the output of these
two invocations.

\subsubsection{Procedure \COVER{}}
The procedure maintains for each vertex $v \in V$, the following variables:
\begin{DenseItemize}

\item
$\EdgeId(v) =$ an identifier $\Id(e)$ of some edge $e \in E$; and

\item
$\Eff(v) =$ a (not necessarily positive) integer refereed to as the
\emph{effectiveness} of $v$.

\end{DenseItemize}

We denote by $\EdgeId_{t}(v)$ and $\Eff_{t}(v)$ the values of 
$\EdgeId(v)$ and $\Eff(v)$,
respectively, at time $t$ (i.e., just before $e_t$ is processed).
Procedure \COVER{} that relies on the following definition is presented in
Algorithm~\ref{algorithm:cover}.

\begin{definition*}[\textbf{level, effectiveness}]
Consider edge $e_{t}$ presented at time $t$ and some subset $T \subseteq
e_{t}$.
The \emph{level} of $T$ at time $t$, denoted $\Lev_{t}(T)$, is defined as
\[
\Lev_{t}(T) = \left\lceil \lg \frac{\Ben(T)}{\Cost(e_{t})} \right\rceil \, .
\]
Subset $T$ is said to be \emph{effective} at time $t$ if for every $v \in T$,
it holds that
\[
\Lev_{t}(T) > \Eff_{t}(v) \, .
\]
\end{definition*}
\noindent
Note that $\emptyset$ is always vacuously effective.

\begin{algorithm}[h]
\caption{\label{algorithm:cover}
$\COVER(G = (V, E, \Ben, \Cost))$}
\begin{algorithmic}

\STATE \textbf{Initialization}
$\forall v \in V$:
$\EdgeId(v) \leftarrow \Null$  and 
$\Eff(v) \leftarrow - \infty$

\FOR{$t = 0, 1, \ldots$}
  \STATE Read edge $e_t \in E$ from the stream
  \STATE Compute an effective subset $T \subseteq e_{t}$ of largest benefit
$\Ben(T)$  \\
  \FORALL{$v \in T$}
    \STATE $\EdgeId(v) \leftarrow \Id(e_{t})$
    \STATE $\Eff(v) \leftarrow \Lev_{t}(T)$
  \ENDFOR
\ENDFOR

\RETURN  $\EdgeId(\cdot)$ and $\Eff(\cdot)$
\end{algorithmic}
\end{algorithm} 

\subsubsection{Algorithm \SSSC{}}
\label{section:algorithm-sssc}
We are now ready to present our algorithm \SSSC{}.
On input weighted graph $G = (V, E, \Ben, \Cost)$,
algorithm \SSSC{} runs in parallel the following procedures that process the
stream of edges:
\begin{DenseEnumerate}

\item[\textbf{P1:}]
$(\EdgeId_{\infty}(\cdot), \Eff_{\infty}(\cdot))
\leftarrow
\COVER(G = (V, E, \Ben, \Cost))$.

\item[\textbf{P2:}]
$(\EdgeId^{\mathbf{1}}_{\infty}(\cdot), \Eff^{\mathbf{1}}_{\infty}(\cdot))
\leftarrow
\COVER(G = (V, E, \mathbf{1}, \Cost))$,
where $\mathbf{1}$ stands for the function that assigns a unit benefit to all
vertices $v \in V$.

\item[\textbf{P3:}]
A procedure that maintains for every vertex $v \in V$, a variable
$\MinEdgeCost(v)$ that stores the identifier of the minimum cost edge that
covers $v$, seen so far.

\item[\textbf{P4:}]
A procedure that stores for every vertex $v \in V$, its benefit $\Ben(v)$.

\end{DenseEnumerate}

Upon termination of the input stream, \SSSC{} takes some parameter
$0 \leq \epsilon < 1$
and extracts the desired $(1 - \epsilon)$-cover certificate for $G$ from the
variables returned by procedures P1--P4.
We distinguish between the following two cases.
\begin{DenseItemize}

\item
Case $\epsilon \geq 1 / \sqrt{n}$: \\
The algorithm looks for the largest integer $r^{*} $ such that
$\Ben(I(\leq r^{*})) \leq \epsilon \Ben(V)$,
where
\[
I(\leq r^{*}) = \{ v \in V : \Eff_{\infty}(v) \leq r^{*} \} \, ,
\]
and returns the partial function
$\chi : V \rightarrow \Id(E)$
that maps every vertex
$v \in V - I(\leq r^{*})$
to
$\EdgeId_{\infty}(v)$.

\item
Case $\epsilon < 1 / \sqrt{n}$: \\
The algorithm looks for the largest integer $r^{*}$ such that
$|I^{\mathbf{1}}(\leq r^{*})| \leq \sqrt{n}$,
where
\[
I^{\mathbf{1}}(\leq r^{*}) = \{ v \in V : \Eff^{\mathbf{1}}_{\infty}(v) \leq
r^{*} \}
\]
and sets $\chi'$ to be the partial function
$\chi' : V \rightarrow \Id(E)$
that maps every vertex
$v \in V - I^{\mathbf{1}}(\leq r^{*})$
to
$\EdgeId^{\mathbf{1}}_{\infty}(v)$.
Then, it returns the (complete) function
$\chi'' : V \rightarrow \Id(E)$
extended from $\chi'$ by mapping every vertex
$v \in I^{\mathbf{1}}(\leq r^{*})$
to
$\MinEdgeCost(v)$.

\end{DenseItemize}

Notice that the unweighted case is much simpler:
If $G = (V, E)$, then procedure P2 is identical to procedure P1;
moreover, procedures P3 and P4 are  redundant since all vertices/edges admit a
unit benefit/cost.
Further note that procedures P1--P4 are oblivious to
$\epsilon$.
Upon termination of the input stream, the algorithm extracts, for the given
$0 \leq \epsilon < 1$,
the desired $(1 - \epsilon)$-cover certificate for $G$ from the variables
returned by procedures P1--P4.  
In fact, several such cover certificates can be extracted for different values
of $\epsilon$.

\subsection{Analysis}
\label{section:algorithm-analysis}
We begin our analysis with some observations regarding our main procedure
\COVER{}.

\begin{observation} \label{observation:adding-vertex-to-effective-subset}
If $T \subseteq e_{t}$ is effective at time $t$ and $v \in T$, then $T \cup \{
u \}$ is effective at time $t$ for every $u \in e_{t}$ such that
$\Eff_{t}(u) \leq \Eff_{t}(v)$.
\end{observation}

Notice that \COVER{}'s updating rule guarantees that the effectiveness
$\Eff(v)$ is non-decreasing throughout the course of the execution.
Employing \Obs{}~\ref{observation:adding-vertex-to-effective-subset}, we
can now derive \Obs{}\ \ref{observation:main-procedure-run-time} and
\ref{observation:effective-ratio-increases} (the former follows by sorting the
vertices $v \in e_{t}$ in non-decreasing order of the value of the
effectiveness $\Eff(v)$).

\begin{observation} \label{observation:main-procedure-run-time}
The run-time of \COVER{} on edge $e_{t}$ is $O (|e_{t}| \log |e_{t}|)$.
\end{observation}

\begin{observation} \label{observation:effective-ratio-increases}
If $T \subseteq e_{t}$ is effective at time $t$, then for every $v \in T$, it
holds that
\[
\Eff_{t + 1}(v)
\geq
\Lev_t(T) \, .
\]
\end{observation}

We are now ready to establish the following lemma.

\begin{lemma} \label{lemma:up-bound-benefit-single-edge}
Consider some integer $r$.
Procedure \COVER{} guarantees that
\[
\Ben \left( \left\{ v \in e_{t} \mid \Eff_{t + 1}(v) \leq r \right\} \right)
<
2^{r + 1} \cdot \Cost(e_{t}) \, .
\]
\end{lemma}
\begin{proof}
Assume by contradiction that there exists a subset $R \subseteq e_{t}$,
$\Ben(R) \geq 2^{r + 1} \cdot \Cost(e_{t})$,
such that
$\Eff_{t + 1}(v) \leq r$
for every $v \in R$.
Since the effectiveness is non-decreasing, it follows that
$\Eff_{t}(v) \leq r$
for every $v \in R$, hence the assumption that
$\Ben(R) \geq 2^{r + 1} \cdot \Cost(e_{t})$
ensures that $R$ is effective at time $t$.
But by \Obs{}~\ref{observation:effective-ratio-increases}, the
effectiveness
$\Eff_{t + 1}(v)$
should have been at least $r + 1$ for every $v \in R$, in contradiction to
the choice of $R$.
\end{proof}

Let $\Eff_{\infty}(v)$ denote the value of the variable $\Eff(v)$ upon
termination of the input stream.
Given some integer $r$, define
\[
I(r) = \left\{ v \in V \mid \Eff_{\infty}(v) = r \right\}
\quad \text{and} \quad
S(r) = \left\{ e \in E \mid \exists v \in I(r) \text{ s.t. } \EdgeId(v)
= \Id(e) \right\}
\]
in accordance with the notation defined in
\Sect{}~\ref{section:algorithm-sssc}.
We extend these two definitions to intervals of integers in the natural way
and denote the intervals
$(-\infty, r]$
and
$(r, \infty)$
in this context by $\leq r$ and $> r$, respectively.

\begin{lemma} \label{lemma:up-bound-benefits}
Consider some integer $r$.
Procedure \COVER{} guarantees that
\[
\Ben(I(\leq r)) < 2^{r+1}\cdot \Cost(\Opt) \, .
\]
\end{lemma}
\begin{proof}
Since the effectiveness is non-decreasing,
\Lem{}~\ref{lemma:up-bound-benefit-single-edge} ensures that for every edge $e
\in E$, it holds that
\[
\Ben \left( \left\{ v \in e \mid \Eff_{\infty}(v) \leq r \right\} \right)
<
2^{r+1} \cdot \Cost(e) \, .
\]
The assertion is established by observing that
\[
\Ben(I(\leq r))
\leq
\sum_{e \in \Opt} \Ben \left( \left\{ v \in e \mid
\Eff_{\infty}(v) \leq r \right\} \right)
<
\sum_{e \in \Opt} 2^{r+1} \cdot \Cost(e)
 =
2 ^{r+1}\cdot \Cost(\Opt) \, ,
\]
where the first inequality is due to the fact that $\Opt$ is an edge cover of
$G$.
\end{proof}

\Lem{}~\ref{lemma:up-bound-benefits} will be used to bound from above the
benefit of the vertices that are not covered by the edges returned by our
algorithm.
We now turn to bound from above the cost of these edges.

\begin{lemma} \label{lemma:up-bound-edge-costs}
Consider some integer $r$.
The edge collection $S(r)$ satisfies
\[
\Cost(S(r))
<
\Ben(V) / 2^{r - 1} \, .
\]
\end{lemma}
\begin{proof}
If $e_{t} \in S(r)$, then there exists some subset $R = R(e_{t}) \subseteq
e_{t}$ with $\Lev_t(R) = r$
such that for every vertex $v \in R$, we have
(1) $\Eff_{t}(v) < r$; and
(2) $\Eff_{t + 1}(v) = r$.
By definition, the fact that $\Lev_t(R) = r$ implies that
$\Cost(e_{t}) < \Ben(R) / 2^{r - 1}$.
Since the variable $\EdgeId(v)$ is updated only when $\Eff(v)$ increases and
since $\Eff(v)$ is non-decreasing, it follows that if $e_{t}, e_{t'} \in
S(r)$, $e_{t} \neq e_{t'}$, then the subsets $R(e_{t})$ and $R(e_{t'})$ are
disjoint.
Therefore,
\[
\sum_{e_{t} \in S(r)} \Cost(e_{t})
<
\frac{1}{2^{r - 1}} \sum_{e_{t} \in S(r)} \Ben(R(e_{t}))
\leq
\Ben(V) / 2^{r - 1}
\]
which completes the proof.
\end{proof}

The following corollary is obtained by applying
\Lem{}~\ref{lemma:up-bound-edge-costs} to the integers $r+1, r+2, \dots$

\begin{corollary} \label{corollary:up-bound-edge-costs}
Consider some integer $r$.
The edge collection $S(> r)$ satisfies
\[
\Cost(S(> r))
<
 \Ben(V) / 2^{r-1} \, .
\]
\end{corollary}

The following important lemma shows that we can extract from the variables
returned by \COVER{} an edge subset of low total cost which covers much of the
items.

\begin{lemma} \label{lemma:main-procedure}
Consider some $0 < \epsilon < 1$ and let $r^*$ be the
largest integer such that
$\Ben(I(\leq r^*)) \leq \epsilon \cdot \Ben(V)$.
The edge collection $S(> r^*)$ satisfies
\[
\Cost(S(> r^*)) < 8 \cdot \Cost(\Opt) / \epsilon \, .
\]
\end{lemma}
\begin{proof}
Let $r$ be an integer such that
$2^{r + 1} < \epsilon \cdot \frac{\Ben(V)}{\Cost(\Opt)} \leq 2^{r + 2}$.
\Lem{}~\ref{lemma:up-bound-benefits} guarantees that
$
\Ben(I(\leq r))
<
2^{r+1} \cdot \Cost(\Opt)
<
\epsilon \cdot \Ben(V) 
$,
hence $r \leq r^*$.
It follows by \Cor{}~\ref{corollary:up-bound-edge-costs} that 
$
\Cost(S(>r^*))
\leq 
\Cost(S(> r))
<
\Ben(V) / 2^{r-1}
\leq 
 8 \cdot   \Cost(\Opt) /  \epsilon 
$.
\end{proof}

We are now ready to establish the approximation guarantees of algorithm
\SSSC{}.
Theorem~\ref{theorem:upper-bound-assumption} (stated under the
assumption that all vertex benefits and edge costs are encoded using $O (\log
n)$ bits) follows immediately from
Theorem~\ref{theorem:algorithm-performance-guarantee}.

\begin{theorem} \label{theorem:algorithm-performance-guarantee}
For any $0 \leq \epsilon < 1$, our algorithm outputs a $(1 - \epsilon)$-cover
certificate for $G$ whose image has cost
$O \left( \min \left\{ \frac{1}{\epsilon}, \sqrt{n} \right\} \cdot
\Cost(\Opt) \right)$.
\end{theorem}
\begin{proof}
If $\epsilon \geq 1 / \sqrt{n}$, then the assertion follows immediately from
\Lem{}~\ref{lemma:main-procedure}, so it remains to consider the case of
$\epsilon < 1 / \sqrt{n}$.
We show that $\chi''$ is a $1$-cover certificates for $G$ such that
$\Cost(\Image(\chi'')) = O (\sqrt{n} \cdot \Cost(\Opt))$.
Observe first that since $\Opt$ covers all vertices in $V$, it is also an
optimal edge $1$-cover of $G^{\mathbf{1}}$.
Thus, \Lem{}~\ref{lemma:main-procedure} guarantees that
$\Cost (\Image(\chi')) < 8 \sqrt{n} \cdot \Cost(\Opt)$.
The vertices $v \in V - \Domain(\chi')$ are mapped under $\chi''$ to
$\MinEdgeCost(v)$.
Since 
$|V - \Domain(\chi')| \leq \sqrt{n}$
and since
$\Cost(\MinEdgeCost(v)) \leq \Cost(\Opt)$
for every $v \in V$, it follows that
\[
\Cost(\Image(\chi''))
<
8 \sqrt{n} \cdot \Cost(\Opt) + |V - \Domain(\chi')| \cdot \Cost(\Opt)
\leq
9 \sqrt{n} \cdot \Cost(\Opt) \, .
\]
The assertion follows.
\end{proof}

\subsection{Lifting the assumption on the numerical values}
\label{section:algorithm-space}
We now turn to lift the assumption that all numerical values are encoded using
$O (\log n)$ bits and establish \Thm{}\ \ref{theorem:upper-bound} and
\ref{theorem:upper-bound-aspect-ratio}, starting with the former.
To that end, consider the hypergraph
$\widetilde{G} = (V, E, \widetilde{\Ben}, \widetilde{\Cost})$
defined by setting
$\widetilde{\Ben}(v) = 2^{\lfloor \lg \Ben(v) \rfloor}$
for every vertex $v \in V$
and
$\widetilde{\Cost}(e) = 2^{\lfloor \lg \Cost(e) \rfloor}$
for every edge $e \in E$.
Since $\widetilde{\Ben}(U)$ and $\widetilde{\Cost}(F)$ are $2$-approximations
of $\Ben(U)$ and $\Cost(F)$, respectively, for every $U \subseteq V$ and $F
\subseteq E$, it follows that a $(1 - O (\epsilon))$-cover certificate for $G$
with image of cost
$O \left( \min \left\{ \frac{1}{\epsilon}, \sqrt{n} \right\} \cdot
\Cost(\Opt) \right)$
can be obtained by running \SSSC{} on $\widetilde{G}$.

So, in what follows, we assume that $\Ben(v)$ and $\Cost(e)$ are (not
necessarily positive) integral powers of $2$ for every vertex $v \in V$ and
edge $e \in E$.
This implies that every benefit $\Ben(v)$ (resp., cost $\Cost(e)$) in $G$ can
be encoded using
$O (\log \Ben^{\lg})$
(resp.,
$O (\log \Cost^{\lg})$)
bits simply by taking the standard binary representation of
$\lg \Ben(v)$
(resp.,
$\lg \Cost(e)$).
Therefore, procedures P3 and P4
can be implemented using
$O \left( \log \left( n + m + \Ben^{\lg} + \Cost^{\lg} \right) \right)$
bits per vertex, as desired.
Procedure \COVER{} can also be implemented with that many bits per vertex
since the level at time $t$ of each subset $T \subseteq e_{t}$ is an integer
whose absolute value satisfies
$|\Lev_{t}(T)| = O (\Ben^{\lg} + \Cost^{\lg} + \log n)$,
thus establishing \Thm{}~\ref{theorem:upper-bound} due to
\Obs{}~\ref{observation:main-procedure-run-time} and
\Thm{}~\ref{theorem:algorithm-performance-guarantee}.

For \Thm{}~\ref{theorem:upper-bound-aspect-ratio}, we need two
additional features.
First, we scale in an online fashion all vertex benefits and edge costs so
that
$\min_{v \in V} \Ben(v)$
and
$\min_{e \in E} \Cost(e)$
are always $1$.
We do the same thing with the effectiveness variables $\Eff(v)$, only that
this time, we ignore those variables with
$\Eff(v) = -\infty$.
This is carried out by maintaining the true values of
$\min_{v \in V} \Ben(v)$,
$\min_{e \in E} \Cost(e)$,
and
$\min_{v \in V : \Eff(v) > -\infty} \Eff(v)$ ---
denote them by 
$\Ben_{\min}$, $\Cost_{\min}$, and $\Eff_{\min}$, respectively ---
and scaling all values of $\Ben(v)$, $\Cost(e)$, and $\Eff(v)$ stored in the
data structures maintained by the procedures of our our algorithm by
$\Ben_{\min}$, $\Cost_{\min}$, and $\Eff_{\min}$, respectively.
Notice that this online scaling requires updating the existing values stored
in the data structures whenever $\Ben_{\min}$, $\Cost_{\min}$, or
$\Eff_{\min}$ are updated, thus resulting in the slightly less favorable
run-time promised by \Thm{}~\ref{theorem:upper-bound-aspect-ratio}.

This online scaling feature ensures that the space allocated for the variables
of each vertex $v$ is now
\begin{equation} \label{equation:temp-space-bound}
O \left( \log \left( n + m + \Ben^{\Delta} + \Cost^{\Delta} \right) \right) \, ,
\end{equation}
where
$\Ben^{\Delta} = \lg \left\lceil \frac{\max_{v \in V} \Ben(v)}{\min_{v \in V}
\Ben(v)} \right\rceil$
is the number of bits required to encode the vertex benefits aspect ratio.
We also need additional
$O (\log (\Ben^{\lg} + \Cost^{\lg}))$
bits to store the variables $\Ben_{\min}$, $\Cost_{\min}$, and
$\Eff_{\min}$.

In order to get rid of the dependency on $\log \Ben^{\Delta}$ in
(\ref{equation:temp-space-bound}) and obtain the space bound promised by
\Thm{}~\ref{theorem:upper-bound-aspect-ratio}, we use the following
feature:
Let
$\sigma = \sum_{v \in V'} \Ben(v)$, where $V'$ is the set of vertices $v \in
V$ encountered by the algorithm so far.
Whenever it becomes clear that the contribution of some vertex $v \in V$ to
$\Ben(V)$
is at most
$\epsilon \cdot \Ben(V) / n$,
which is indicated by
$\Ben(v) \leq \epsilon \sigma / n$,
the algorithm marks vertex $v$ as \emph{insignificant}.
Insignificant vertices are treated as if they are not part of the input
hypergraph $G$;
in particular, upon marking vertex $v$ as insignificant, the algorithm erases
any variable associated with $v$ and updates $\Ben_{\min}$ so that it does not
take $\Ben(v)$ into account.

Notice that the total contribution of all insignificant vertices to $\Ben(V)$
is bounded from above by
$\epsilon \cdot \Ben(V)$.
Therefore, ignoring insignificant vertices cannot hurt our guaranteed coverage
by more than an additive term of $\epsilon \cdot \Ben(V)$.
The key observation now is that by ignoring insignificant vertices, we keep
the parameter $\Ben^{\Delta}$ bounded by
$\Ben^{\Delta} = O (\log (n / \epsilon))$
as the benefit of any vertex encountered by the algorithm so far is
clearly at most $\sigma$.
Recalling that $\epsilon$ is always at least $1 / \sqrt{n}$, we conclude that
the dependency on $\log \Ben^{\Delta}$ in (\ref{equation:temp-space-bound})
is replaced by a dependency on
$\log\log n$.
\Thm{}~\ref{theorem:upper-bound-aspect-ratio} follows by
\Thm{}~\ref{theorem:algorithm-performance-guarantee}.

\section{Lower bounds}
A \emph{randomized} semi-streaming algorithm \Alg{} for the edge cover problem
in hypergraphs is said to be an
\emph{$(n, s, \epsilon, \rho)$-algorithm}
(resp., an \emph{uncertified $(n, s, \epsilon, \rho)$-algorithm})
if given any $n$-vertex unweighted hypergraph $G$,
\Alg{} is guaranteed to maintain a memory of size at most $s$ bits and
to output a $(1 - \epsilon)$-cover certificate for $G$ with image of expected
cardinality at most $\rho \cdot |\Opt|$
(resp., to output the identifiers of an edge $(1 - \epsilon)$-cover of $G$
whose expected cardinality is at most $\rho \cdot |\Opt|$),
where $\Opt$ is an optimal edge cover of $G$.
Our goal in this section is to establish \Thm{}\ \ref{theorem:lower-bound}
and \ref{theorem:lower-bound-uncertified}, treated in \Sect{}\
\ref{section:certified-case} and \ref{section:uncertified-case},
respectively.
Observe that the constructions that lie at the heart of Theorems
\ref{theorem:lower-bound} and \ref{theorem:lower-bound-uncertified} are based
on hypergraphs whose number of vertices and number of edges are polynomially
related, that is, $m = n^{\Theta (1)}$.

\begin{theorem} \label{theorem:lower-bound}
For every integer $n_{0}$,
there exists an integer $n \geq n_{0}$ such that
for every
$\epsilon = \Omega (1 / \sqrt{n})$,
the existence of an
$(n, o (n^{3 / 2}), \epsilon, \rho)$-algorithm
implies that
$\rho = \Omega (1 / \epsilon)$.
\end{theorem}

\begin{theorem} \label{theorem:lower-bound-uncertified}
Fix some constant real $\alpha > 0$.
For every integer $n_{0}$,
there exists an integer $n \geq n_{0}$ such that
for every
$\epsilon \geq n^{-1 / 2 + \alpha}$,
the existence of an uncertified
$(n, o (n^{1 + \alpha}), \epsilon, \rho)$-algorithm
implies that
$\rho = \Omega \left( \frac{\log\log n}{\log n} \frac{1}{\epsilon} \right)$.
\end{theorem}

\subsection{The certified case}
\label{section:certified-case}
We shall establish \Thm{}~\ref{theorem:lower-bound} by introducing a
probability distribution $\mathcal{G}$ over $n$-vertex hypergraphs that
satisfy the following two properties:
(1) Every hypergraph in the support of $\mathcal{G}$ admits an edge cover of
cardinality $O (\epsilon \sqrt{n})$.
(2) For every \emph{deterministic} semi-streaming algorithm \Alg{} that given
an $n$-vertex hypergraph $G$, maintains a memory of size
$o (n^{3 / 2})$
and outputs a $(1 - \epsilon)$-cover certificate $\chi$ for $G$,
when \Alg{} is invoked on a hypergraph chosen according to $\mathcal{G}$, the
expected cardinality of $\Image(\chi)$ is $\Omega (\sqrt{n})$.
The theorem than follows by Yao's principle.

\subsubsection{The construction of $\mathcal{G}$}
\label{section:lower-bound-construction}
Let $q$ be a large prime power.
Our construction relies on the \emph{affine plane} $\mathcal{A}
= (P, L)$, where $P$ is a set of $q^{2}$ \emph{points} and $L \subseteq 2^{P}$
is a set of $q (q + 1)$ \emph{lines} satisfying the following
properties: \\
(1) every line contains $q$ points; \\
(2) every point is contained in $q + 1$ lines; \\
(3) for every two distinct points, there is exactly one line that contains
both of them; and \\
(4) every two lines intersect in at most one point. \\
Two lines with an empty intersection are called \emph{parallel}.
The line set $L$ can be partitioned into $q + 1$ clusters
$A_{1}, \dots, A_{q + 1}$
referred to as \emph{angles}, where
$A_{i} = \{ \ell_{i}^{1}, \dots, \ell_{i}^{q} \}$
for $i = 1, \dots, q + 1$, such that two distinct lines are parallel if and
only if they belong to the same angel.
Refer to \cite{LindnerR2011} for an explicit construction of such a
combinatorial structure.

Consider some
$\frac{1}{3 q} \leq \epsilon \leq \frac{1}{66} - \frac{1}{3 q}$
and let
$r = \lceil 3 \epsilon q \rceil$.
We construct a random hypergraph $G = (V, E)$ based on the affine plane
$\mathcal{A} = (P, L)$ as follows (refer to Figure~\ref{figure:bad-hypergraph}
for an illustration).
Fix $V = P$.
Randomly partition each line $\ell \in L$ into $2$ edges
$e_{1}(\ell) \cup e_{2}(\ell) = \ell$
by assigning each point in $L$ to one of the $2$ edges u.a.r.\ (and
independently of all other random choices).\footnote{
Throughout, we use u.a.r.\ to abbreviate ``uniformly at random''.
}
It will be convenient to denote the set of edges corresponding to the lines in
angle $A_{i}$ by 
$E_{i} = \{ e_{1}(\ell), e_{2}(\ell) \mid \ell \in A_{i} \}$.
Let
\[
e^{*} = P - \bigcup_{t = 1}^{r} \ell_{i}^{j(t)} \, ,
\]
where
$i$ is an index chosen u.a.r.\ (and independently) from $[q + 1]$ and
$1 \leq j(1) < \cdots < j(r) \leq q$
are $r$ distinct indices chosen u.a.r.\ (and independently) from $[q]$.
In other words, $e^{*}$ is constructed by randomly choosing an angle $A_{i}$
and then randomly choosing $r$ distinct lines
$\ell_{i}^{j(1)}, \dots, \ell_{i}^{j(r)}$
from $A_{i}$;
the edge consists of all points except those contained in these $r$ lines.

\begin{figure}
\begin{center}
\begin{subfigure}[b]{0.3\textwidth}
\begin{center}
\includegraphics[width=\textwidth]{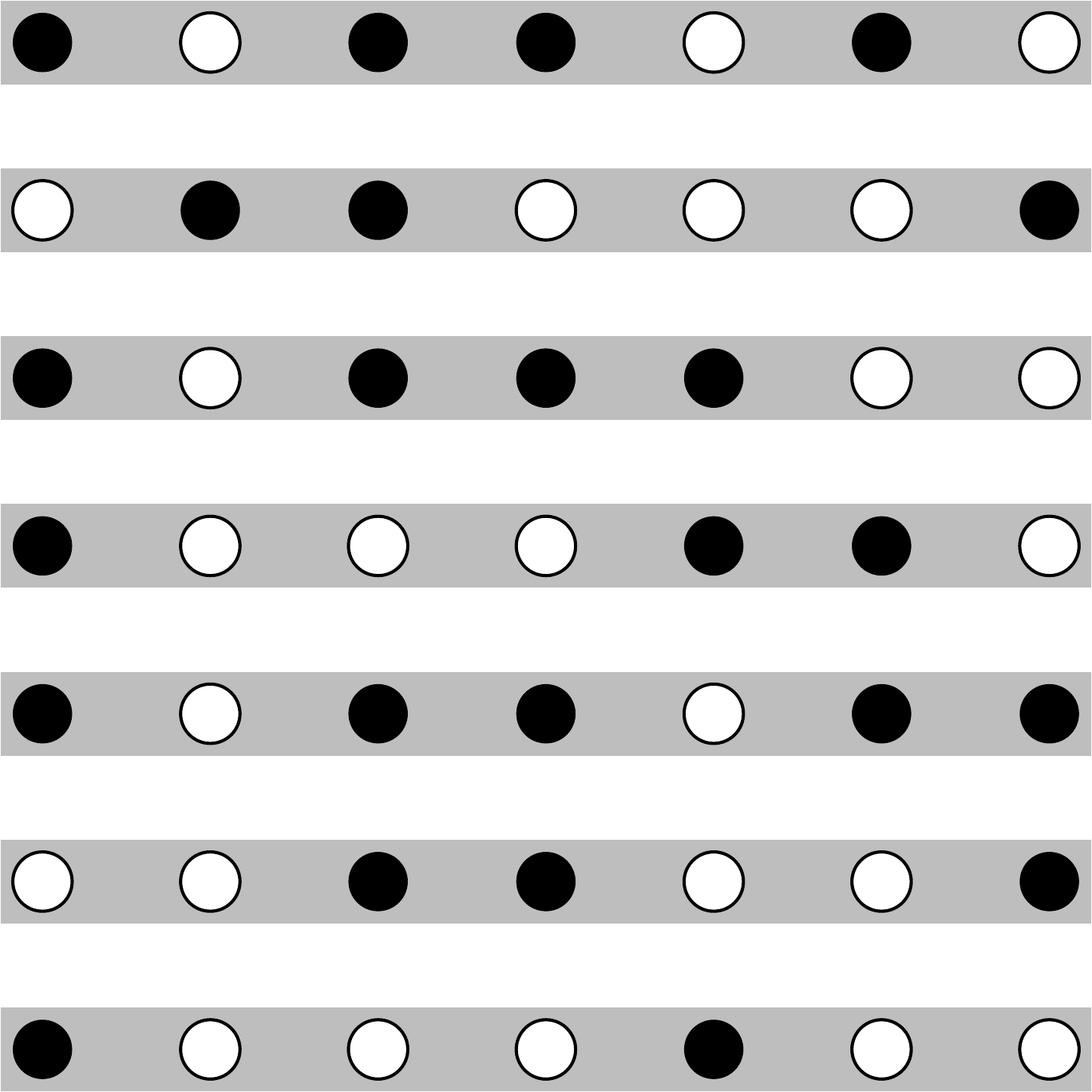}
\caption{The edges in $E_{i}$}
\label{figure:bad-hypergraph:small-edges}
\end{center}
\end{subfigure}
\qquad
\begin{subfigure}[b]{0.3\textwidth}
\begin{center}
\includegraphics[width=\textwidth]{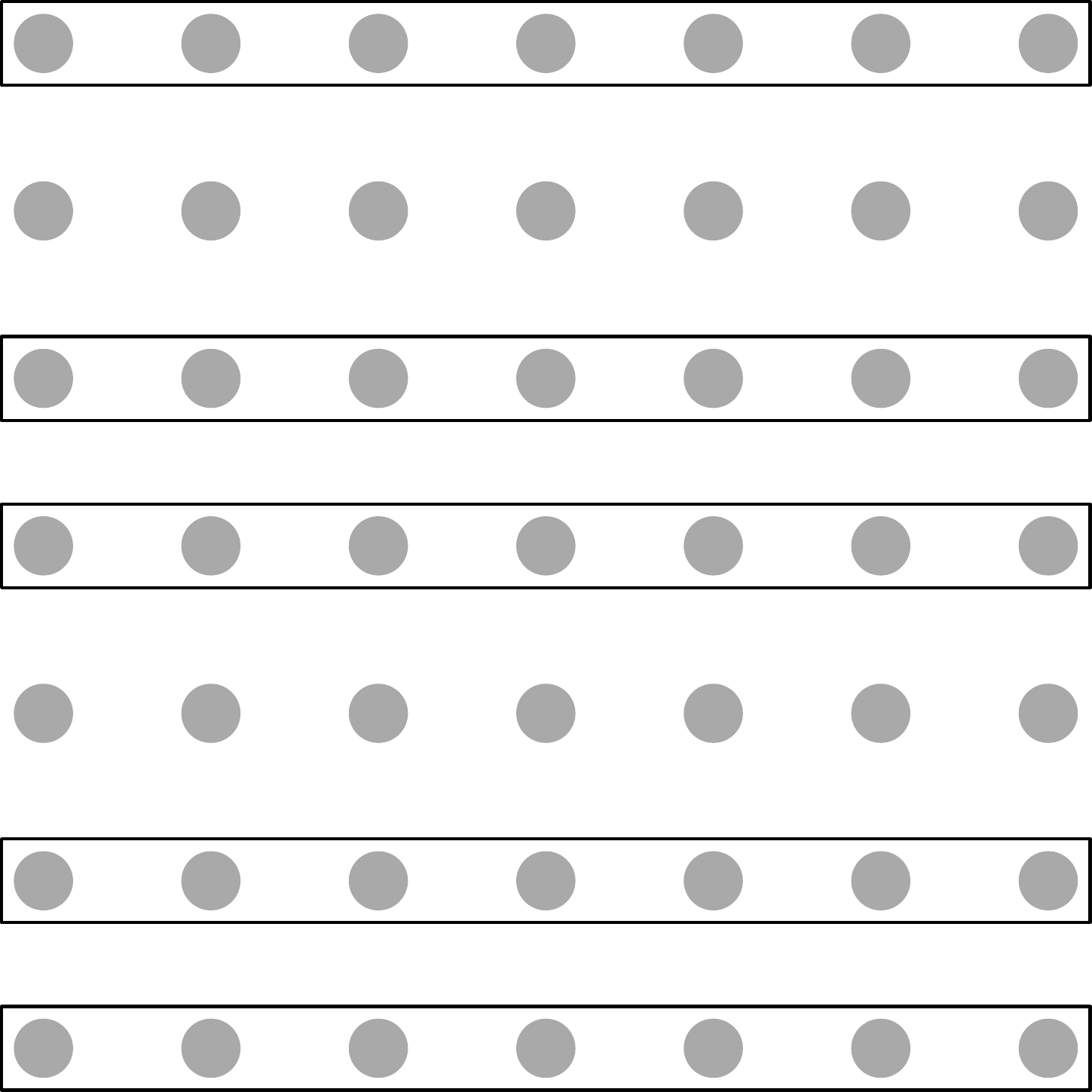}
\caption{Edge $e^{*}$}
\label{figure:bad-hypergraph:large-edge}
\end{center}
\end{subfigure}
\end{center}
\caption{
The hypergraph $G$ for $q = 7$.
(The requirements on $\epsilon$ actually imply that $q$ must be larger, but we
set $q = 7$ for the sake of a clearer illustration.)
The gray rectangles in (\subref{figure:bad-hypergraph:small-edges}) depict the
$7$ parallel lines in angle $A_{i}$ for some $i \in [q + 1]$, whereas the
black/white circles in each line $\ell_{i}^{j}$ depict the points in
$e_{1}(\ell_{i}^{j})$/$e_{2}(\ell_{i}^{j})$.
Edge $e^{*}$, depicted by the white rectangles in
(\subref{figure:bad-hypergraph:large-edge}), consists of all points except
those in $r = 2$ lines of angle $A_{i}$.
}
\label{figure:bad-hypergraph}
\end{figure}

Fix
\[
E = E_{1} \cup \cdots \cup E_{q + 1} \cup \{ e^{*} \} \, .
\]
Observe that
$n = |P| = q^{2}$
and
$m = 1 + 2 \cdot |L| = 1 + 2 \cdot q (q + 1)$.
The execution is divided into two stages, where in the first stage, the edges in
$E_{1} \cup \cdots \cup E_{q + 1}$ are presented in an arbitrary order and in
the second stage, edge $e^{*}$ is presented.

\subsubsection{Analysis}
\label{section:lower-bound-analysis}
We start the analysis by observing that $G$ can be covered by the
edge $e^{*}$ and the edges in
$\{ e_{1}(\ell_{i}^{j(t)}), e_{2}(\ell_{i}^{j(t)}) \mid 1 \leq t \leq r \}$.
Therefore,
\begin{equation} \label{equation:opt-cardinality}
|\Opt|
\leq
2 r + 1
= O (\epsilon q) \, ,
\end{equation}
where the equation follows from the definition of
$r = \lceil 3 \epsilon q \rceil$
due to the requirement that
$\epsilon \geq \frac{1}{3 q}$.

Let $s$ be the space of the deterministic semi-streaming algorithm \Alg{}.
\Thm{}~\ref{theorem:lower-bound} is established by combining
(\ref{equation:opt-cardinality}) with the following lemma
(that ensures an $\Omega (q)$ expected image cardinality whenever $s = o (n^{3
/ 2})$).

\begin{lemma} \label{lemma:lower-bound-exact-parameters}
If
$s \leq q^{2} (q + 1) / 48$,
then w.p.\ $\geq 1 / 8$, the $(1 - \epsilon)$-cover certificate returned by
\Alg{} has image of cardinality at least $q / 3$.\footnote{
Throughout, we use w.p.\ and w.h.p.\ to abbreviate ``with probability'' and
``with high probability'', respectively.
}
\end{lemma}

\paragraph{Bounding the expected entropy.}
The proof of \Lem{}~\ref{lemma:lower-bound-exact-parameters} is based on
information theoretic arguments that require the following definitions.
Let $X_{i}^{j}$ be a random variable that depicts the partition
$(e_{1}(\ell_{i}^{j}), e_{2}(\ell_{i}^{j}))$
of line
$\ell_{i}^{j} = e_{1}(\ell_{i}^{j}) \cup e_{2}(\ell_{i}^{j})$
for every $i \in [q + 1]$ and $j \in [q]$.
Let
$X_{i} = (X_{i}^{1}, \dots, X_{i}^{q})$
and
$X = (X_{1}, \dots, X_{q + 1})$.
The independent random choices in the construction of the hypergraph $G$
guarantee that
$\Entropy(X_{i}^{j}) = q$,
$\Entropy(X_{i}) = q^2$, and
$\Entropy(X) = q^{2} (q + 1)$,
where $\Entropy(\cdot)$ denotes the binary entropy function.
Before we can proceed with our proof, we have to establish the following
lemma whose restriction to the case $k = 1$ is a basic fact in information
theory.
It will not strike us as a surprise if this lemma was already proved beforehand
although we are unaware of any such specific proof;
for the sake of completeness, we provide a full proof of this lemma based on
Baranyai's Theorem in
Appendix~\ref{appendix:proof-lemma-expected-entropy}.

\begin{lemma} \label{lemma:expected-entropy}
Let $X_{1}, \dots, X_{n}, Y$ be $n + 1$ arbitrary random variables and let
$1 \leq j(1) < \cdots < j(k) \leq n$
be $1 \leq k \leq n$ distinct indices chosen u.a.r.\ from $[n]$.
Then,
\[
\left\lceil \frac{n}{k} \right\rceil \Expectation_{j(1), \dots, j(k)} \left[
\Entropy \left( X_{j(1)}, \dots, X_{j(k)} \mid Y \right)
\right]
~ \geq ~
\Entropy \left( X_{1}, \dots, X_{n} \mid Y \right) \, .
\]
\end{lemma}

Let $M$ be a random variable that depicts the memory image of \Alg{} upon
completion of the first stage of the execution.
Since $M$ is fully determined by $X$, it follows that
$\Entropy(X, M) = \Entropy(X)$,
hence
$\Entropy(X \mid M)
=
\Entropy(X) - \Entropy(M)$.
Recalling that $M$ is described by $s$ bits, we conclude that
$\Entropy(M) \leq s \leq q^{2} (q + 1) / 48$,
thus
\begin{equation} \label{equation:entropy-in-X}
\Entropy(X \mid M)
\geq
\frac{47}{48} \cdot q^{2} (q + 1)
=
\frac{47}{48} \cdot \Entropy(X) \, .
\end{equation}
We are now ready to establish the following lemma.

\begin{lemma} \label{lemma:entropy-in-r-lines}
Our construction guarantees that
\[
\Probability_{i, j(1), \dots, j(r)} \left( \Entropy \left( X_{i}^{j(1)},
\dots, X_{i}^{j(r)} \mid M \right)
\geq
\frac{5}{6} \cdot r q \right) \geq 1 / 4 \, ,
\]
where $i \in [q + 1]$ and
$1 \leq j(1) < \cdots < j(r) \leq q$
are the random indices chosen during the construction of edge $e^{*}$.
\end{lemma}
\begin{proof}
By combining (\ref{equation:entropy-in-X}) with an application of
\Lem{}~\ref{lemma:expected-entropy} to the random choice of index $i \in [q +
1]$, we derive the inequality
\[
\Expectation_{i} \left[ \Entropy \left( X_{i} \mid M \right) \right]
\geq
\frac{47}{48} \cdot q^{2} \, .
\]
Since
$\Entropy(X_{i} \mid M) \leq q^{2}$,
we can apply Markov's inequality to conclude that
\begin{equation} \label{equation:condition-event-half}
\Entropy(X_{i} \mid M)
\geq
\frac{23}{24} \cdot q^{2}
\end{equation}
w.p.\ $\geq 1 / 2$.

Conditioned on the event that (\ref{equation:condition-event-half}) holds, we
can apply \Lem{}~\ref{lemma:expected-entropy} to the random choice of indices
$1 \leq j(1) < \cdots < j(r) \leq q$,
deriving the inequality
\[
\left\lceil \frac{q}{r} \right\rceil \Expectation_{j(1), \dots, j(r)} \left[
\Entropy \left( X_{i}^{j(1)}, \dots, X_{i}^{j(r)} \mid M \right) \right]
\geq
\frac{23}{24} \cdot q^{2}
\]
which means that
\[
\Expectation_{j(1), \dots, j(r)} \left[ \Entropy \left( X_{i}^{j(1)}, \dots, X_{i}^{j(r)}
\mid M \right) \right]
\geq
\frac{23}{24} \frac{r q^2}{q + r} \, .
\]
Since
$\epsilon \leq \frac{1}{66} - \frac{1}{3 q}$,
it follows that
$r
=
\lceil 3 \epsilon q \rceil
\leq
3 \epsilon q + 1
\leq
q / 22$.
This, in turn, implies that
$\frac{23}{24} \frac{r q^2}{q + r} \geq \frac{11}{12} r q$
which guarantees that
\[
\Expectation_{j(1), \dots, j(r)} \left[
\Entropy \left( X_{i}^{j(1)}, \dots, X_{i}^{j(r)} \mid M \right) \right]
\geq
\frac{11}{12} \cdot r q \, .
\]
Since
$\Entropy(X_{i}^{j(1)}, \dots, X_{i}^{j(r)} \mid M) \leq r q$,
we can apply Markov's inequality to conclude that
\[
\Entropy \left( X_{i}^{j(1)}, \dots, X_{i}^{j(r)} \mid M \right)
\geq
\frac{5}{6} \cdot r q
\]
w.p.\ $\geq 1 / 2$.
The assertion follows as (\ref{equation:condition-event-half}) holds w.p.\
$\geq 1 / 2$.
\end{proof}

\paragraph{Introducing the random variable $Z$.}
Let $\mu$ be the actual memory image of \Alg{} upon completion of the first
stage of the execution and recall that $\mu$ is some instance of the random
variable $M$.
Let $Z$ be a real valued random variable that maps the event $M = \mu$ to the
entropy in the joint random variable $X_{i}^{j(1)}, \dots, X_{i}^{j(r)}$ given
$M = \mu$.
Observe that by the definition of conditional entropy, we have
$\Expectation[Z] = \Entropy(X_{i}^{j(1)}, \dots, X_{i}^{j(r)} \mid M)$.
If the event described in \Lem{}~\ref{lemma:entropy-in-r-lines} occurs, then
$\Expectation[Z]
\geq
\frac{5}{6} \cdot r q$
and since $Z$ is never larger than $r q$, we can apply Markov's inequality to
conclude that
\[
\Entropy \left( X_{i}^{j(1)}, \dots, X_{i}^{j(r)} \mid M = \mu \right)
\geq
\frac{2}{3} \cdot r q
\]
w.p.\ $\geq 1 / 2$.
The following corollary is established since the event described in
\Lem{}~\ref{lemma:entropy-in-r-lines} holds w.p.\ $\geq 1 / 4$.

\begin{corollary} \label{corollary:large-entropy-remains}
W.p.\ $\geq 1 / 8$, the entropy that remains in
$X_{i}^{j(1)}, \dots, X_{i}^{j(r)}$
after $e^{*}$ is exposed to \Alg{} given that $M = \mu$ is at least
$\frac{2}{3} \cdot r q$ bits.
\end{corollary}

\paragraph{High entropy implies a large edge cover.}
Condition hereafter on the event described in
\Cor{}~\ref{corollary:large-entropy-remains}.
Consider the $(1 - \epsilon)$-cover certificate $\chi$ returned by \Alg{} and
let
$P' = \bigcup_{t = 1}^{r} \ell_{i}^{j(t)} = P - e^{*}$
be the set of points not covered by $e^{*}$.
Let
\[
R = \left\{ p \in P' \mid p \in \Domain(\chi) \land \chi(p) \in E_{i}
\right\}
\]
be the set of points not covered by $e^{*}$ that are mapped under $\chi$ to
some edge in $E_{i}$, where recall that $E_{i}$ is the set of edges
corresponding to the lines in angle $A_{i}$ (the angle chosen in the random
construction of $e^{*}$).
We can now establish the following lemma.

\begin{lemma} \label{lemma:missing-points}
Our construction guarantees that
$|R| \leq r q / 3$.
\end{lemma}
\begin{proof}
The joint random variable
$X_{i}^{j(1)}, \dots, X_{i}^{j(r)}$
conditioned on $M = \mu$ can be viewed as a probability distribution $\pi$
over the matrices $T \in \{ 1, 2 \}^{r \times q}$, where $T(t, k) \in \{ 1, 2 \}$
indicates whether the $k^{\text{th}}$ point in line $\ell_{i}^{j(t)}$ belongs
to edge
$e_{1}(\ell_{i}^{j(t)})$
or
$e_{2}(\ell_{i}^{j(t)})$
for every $k \in [q]$ and $1 \leq t \leq r$.
Consider some point $p \in R$ and suppose that this is the $k^{\text{th}}$
point in line $\ell_{i}^{j(t)}$.
By the definition of $R$, all matrices $T$ in the support of $\pi$
must agree on $T(t, k)$.\footnote{
In fact, even if we relax the requirement from \Alg{} so that $\chi$ is
allowed to err on some vertices in its domain and the coverage is measured
with respect to the vertices for which $\chi$ is correct, we can still achieve
the desired (asymptotic) bound by using a line of arguments similar to that
used in the proof of Lemma 6.2 in \cite{AlonEFT2013}.
}
Therefore, the entropy that remains in
$X_{i}^{j(1)}, \dots, X_{i}^{j(r)}$
can only arrive from points in $P' - R$.
The assertion follows by \Cor{}~\ref{corollary:large-entropy-remains} since
each such point contributes at most $1$ bit of entropy.
\end{proof}

The cardinality of $\Domain(\chi)$ is at least
$|\Domain(\chi)| \geq (1 - \epsilon) q^{2}$.
The choice of
$r = \lceil 3 \epsilon q \rceil$
ensures that
$\epsilon q^{2} \leq r q / 3$,
thus
$|\Domain(\chi)| \geq q^{2} - r q / 3$.
The key observation now is that even if all these $r q / 3$ missing points
from $\Domain(\chi)$ are in $P'$, it still leaves us with
$|\Domain(\chi) \cap (P' - R)| \geq r q / 3$
by \Lem{}~\ref{lemma:missing-points}.

Every point in $\Domain(\chi) \cap (P' - R)$ is covered by some edge $e \in
E_{j}$, $j \neq i$.
The properties of the affine plane guarantee that each such edge $e$ covers at
most one point in line $\ell_{i}^{j(t)}$, which sums up to at most $r$ points
in $P'$.
Thus, the image of $\chi$ must contain (the identifiers of) at least $q /
3$ different edges.
This concludes the proof of \Lem{}~\ref{lemma:lower-bound-exact-parameters}.
\Thm{}~\ref{theorem:lower-bound} then follows by combining
(\ref{equation:opt-cardinality}) and
\Lem{}~\ref{lemma:lower-bound-exact-parameters}.

\subsection{The uncertified case}
\label{section:uncertified-case}
Similarly to the proof of \Thm{}~\ref{theorem:lower-bound}, we shall
establish \Thm{}~\ref{theorem:lower-bound-uncertified} by introducing a
probability distribution $\mathcal{G}'$ over $n$-vertex hypergraphs that this
time, satisfies the following two properties:
(1) Every hypergraph in the support of $\mathcal{G}'$ admits an edge cover of
cardinality $O (\epsilon^{2} n)$.
(2) For every \emph{deterministic} semi-streaming algorithm \Alg{} that given
an $n$-vertex hypergraph $G = (V, E)$, maintains a memory of size
$o (n^{1 + \alpha})$
and outputs the identifiers of an edge $(1 - \epsilon)$-cover $F
\subseteq E$ of $G$,
when \Alg{} is invoked on a hypergraph chosen according to $\mathcal{G}'$, the
expected cardinality of $F$ is
$\Omega \left( \epsilon n \frac{\log\log n}{\log n} \right)$.
The theorem than follows by Yao's principle.

\subsubsection{The construction of $\mathcal{G}'$}
We construct a random hypergraph $\hat{G} = (\hat{V}, \hat{E})$ as follows.
Let $q$ be a large power of $2$ and fix some constant real $\alpha > 0$.
Consider some
$q^{-(1 - \alpha)} \leq \epsilon \leq \frac{1}{66} - \frac{1}{3 q}$
and let
$r = \lceil 3 \epsilon q \rceil$.
The main building block of $\hat{G}$ is very similar to the random hypergraph
$G = (V, E)$ constructed in \Sect{}~\ref{section:lower-bound-construction}
based on the affine plane $\mathcal{A} = (P, L)$.
Specifically, fix $\hat{V} = P$ and let $E'$ be a random edge set constructed
just like the construction of the random edge set $E$ presented in
\Sect{}~\ref{section:lower-bound-construction} with the following exception:
Instead of randomly partitioning each line $\ell \in L$ into $2$ edges
$e_{1}(\ell) \cup e_{2}(\ell) = \ell$
by assigning each point in $L$ to one of the $2$ edges u.a.r.\ (and
independently),
we randomly partition each line $\ell \in L$ into $r$ edges
$e_{1}(\ell) \cup \cdots \cup e_{r}(\ell) = \ell$
by assigning each point in $L$ to one of the $r$ edges u.a.r.\ (and
independently).

The edge $e^{*}$ is constructed in the same manner as in
\Sect{}~\ref{section:lower-bound-construction}, i.e.,
we choose an angle $A_{i}$ u.a.r.\ and then choose $r$ distinct lines
$\ell_{i}^{j(1)}, \dots, \ell_{i}^{j(r)}$
u.a.r.\ from $A_{i}$;
the edge consists of all points except those contained in these $r$ lines.
(Notice that the parameter $r$ is now used for both the partition of each line
into $r$ edges and the construction of edge $e^{*}$.)
For every $i \in [q + 1]$, denote the set of edges corresponding to the lines
in angle $A_{i}$ by
$E'_{i} = \{ e_{1}(\ell), \dots, e_{r}(\ell) \mid \ell \in A_{i} \}$
and fix
$E' = E'_{1} \cup \cdots \cup E'_{q + 1} \cup \{ e^{*} \}$.

The edge multi-set $\hat{E}$ is obtained from $E'$ by augmenting it with
\emph{dummy} edges:
fix $\hat{E} = E' \cup E_{\Dummy}$, where the edges $e \in E_{\Dummy}$,
referred to as dummy edges, are all empty $e = \emptyset$.
(Concerns regarding the usage of empty edges can be lifted by augmenting
$\hat{V}$ with a dummy vertex $v_{\Dummy}$ and taking all dummy edges $e \in
E_{\Dummy}$ to be singletons $e = \{ v_{\Dummy} \}$.)

\paragraph{Identifier assignment.}
Recall that the arrival order of the edges is determined by their identifiers
so that the edge $e_{t}$ arriving at time $t$ is assigned with identifier
$\Id(e_{t}) = t$.
In contrast to the construction presented in
\Sect{}~\ref{section:lower-bound-construction}, where the identifier
assignment is arbitrary (with the exception that $\Id(e^{*})$ should be the
largest identifier), the assignment of identifiers to the edges in $\hat{E}$
plays a key role in the current construction.
Specifically, for every $i \in [q + 1]$, $j \in [q]$, and $k \in [r]$, the
identifier assigned to edge
$e_{k}(\ell_{i}^{j})$
is
\[
\Id(e_{k}(\ell_{i}^{j}))
=
0 \circ i \circ j \circ k \circ X_{i}^{j, k} \, ,
\]
where $i$, $j$, and $k$ are assumed to be encoded as bitstrings of lengths
$\lceil \lg (q + 1) \rceil$, $\lg q$ (recall that $q$ is a power of $2$), and
$\lceil \lg r \rceil$, respectively,
$\circ$ denotes the string concatenation operator, and
$X_{i}^{j, k}$
is a bitstring of length $3 \lg q$ chosen u.a.r.\ (and
independently).
Notice that each identifier contains
$\iota =
1 + \lceil \lg (q + 1) \rceil + \lg q + \lceil \lg r \rceil + 3 \lg q$
bits encoding some integer (with the most significant bit on the
left) in
$[0, 2^{\iota - 1} - 1]$
and by design, each edge in
$E'_{1} \cup \cdots \cup E'_{q + 1}$
is assigned with a unique identifier.

The identifier assigned to edge $e^{*}$ is
$\Id(e^{*}) = 1 \circ 0^{\iota - 1}$,
which encodes the integer $2^{\iota - 1}$.
The dummy edges are used for filling up the gaps between the identifiers
assigned to the edges in $E'$ so that $\Id(\cdot)$ is a bijection from
$\hat{E} = E' \cup E_{\Dummy}$
to
$\left[ 0, 2^{\iota - 1} \right]$.
As $e^{*}$ is assigned with the highest identifier, this is the last edge to
arrive.
Observe that
$n = q^{2}$
and
$m =
2^{\iota - 1} + 1 = 
O (q^{6})$.

\subsubsection{Analysis}
\sloppy
We start the analysis by observing that $\hat{G}$ can be covered by edge
$e^{*}$ and the edges in
$\{ e_{1}(\ell_{i}^{j(t)}), \dots, e_{r}(\ell_{i}^{j(t)}) \mid 1 \leq t \leq r
\}$.
Therefore,
\begin{equation} \label{equation:opt-cardinality-uncertified}
|\Opt|
\leq
r^{2} + 1
= O (\epsilon^{2} q^{2}) \, ,
\end{equation}
where the equation follows from the definition of
$r = \lceil 3 \epsilon q \rceil$
due to the requirement that
$\epsilon = \omega (q^{-1})$.
\par\fussy

Let $s$ be the space of the deterministic semi-streaming algorithm \Alg{}.
\Thm{}~\ref{theorem:lower-bound-uncertified} is established by combining
(\ref{equation:opt-cardinality-uncertified}) with the following lemma
(that ensures an $\tilde{\Omega} ( \epsilon q^{2} )$ expected set cover
cardinality whenever $s = o (\epsilon n^{3 / 2})$).

\begin{lemma} \label{lemma:lower-bound-exact-parameters-uncertified}
If
$s \leq r q (q + 1) / 16$,
then w.p.\ $\geq 1 / 9$, the edge $(1 - \epsilon)$-cover returned by \Alg{}
has cardinality
$\Omega \left( \epsilon q^{2} \frac{\log\log q}{\log q} \right)$.
\end{lemma}

The proof of \Lem{}~\ref{lemma:lower-bound-exact-parameters-uncertified} is
based on information theoretic arguments that require the following
definitions.
Recall that $X_{i}^{j, k}$ is a random bitstring of length $3 \lg q$ used in
the construction of
$\Id(e_{k}(\ell_{i}^{j}))$
for every $i \in [q + 1]$, $j \in [q]$, and $k \in [r]$.
Let
$X_{i}^{j} = (X_{i}^{j, 1}, \dots, X_{i}^{j, r})$,
$X_{i} = (X_{i}^{1}, \dots, X_{i}^{q})$,
and
$X = (X_{1}, \dots, X_{q + 1})$.
The independent random choices in the construction of the identifiers of
$\hat{E}$ guarantee that
$\Entropy(X_{i}^{j, k}) = 3 \lg q$,
$\Entropy(X_{i}^{j}) = 3 r \lg q$,
$\Entropy(X_{i}) = 3 r q \lg q$, and
$\Entropy(X) = 3 r q (q + 1) \lg q$.

As in the analysis performed in \Sect{}~\ref{section:lower-bound-analysis},
let
$i \in [q + 1]$
and
$1 \leq j(1) < \cdots < j(r) \leq q$
be the random indices chosen in the construction of edge $e^{*}$.
Let $M$ be a random variable that depicts the memory image of \Alg{} before
the last edge $e^{*}$ arrives and let $\mu$ be its actual instantiation.
Observing that
$\Entropy(X \mid M) \geq \frac{47}{48} \cdot \Entropy(X)$
(cf.\ inequality~(\ref{equation:entropy-in-X})),
we can repeat the line of arguments used in
\Sect{}~\ref{section:lower-bound-analysis} to derive the following corollary
(analogous to \Cor{}~\ref{corollary:large-entropy-remains}).

\begin{corollary} \label{corollary:large-entropy-remains-uncertified}
W.p.\ $\geq 1 / 8$, the entropy that remains in
$X_{i}^{j(1)}, \dots, X_{i}^{j(r)}$
after $e^{*}$ is exposed to \Alg{} given that $M = \mu$ is at least
$2 r^{2} \lg q$ bits.
\end{corollary}

Notice that the requirement
$\epsilon \geq q^{-(1 - \alpha)}$
ensures that
$r = \lceil 3 \epsilon q \rceil$
and $q$ are polynomially related and so are $r$ and $n = q^{2} + 1$.
Therefore, an event that holds w.h.p.\ with respect to the parameter $r$ also
holds w.h.p.\ with respect to the parameters $q$ and $n$;
in what follows, whenever we use the term w.h.p., we refer to w.h.p. with
respect to these three parameters. 

\begin{lemma} \label{lemma:edge-cardinalities}
W.h.p., all edges
$e_{k}(\ell_{i}^{j(t)})$,
$t \in [r], k \in [r]$,
satisfy
$(5 / 6) q / r \leq |e_{k}(\ell_{i}^{j(t)})| \leq 2 q / r$.
\end{lemma}
\begin{proof}
Fix some $t \in [r]$ and $k \in [r]$.
The random partition of line
$\ell_{i}^{j(t)}$
into the $r$ edges
$e_{1}(\ell_{i}^{j(t)}) \cup \cdots \cup e_{r}(\ell_{i}^{j(t)}) =
\ell_{i}^{j(t)}$
implies that
$\Expectation[|e_{k}(\ell_{i}^{j(t)})|] = q / r$.
By Chernoff's bound, we have
$(5 / 6) q / r \leq |e_{k}(\ell_{i}^{j(t)})| \leq 2 q / r$
w.h.p.
The assertion follows by union bound.
\end{proof}

\paragraph{Identifiers with large entropy.}
Condition hereafter on the events described in
\Cor{}~\ref{corollary:large-entropy-remains-uncertified} and
\Lem{}~\ref{lemma:edge-cardinalities}.
Since \Cor{}~\ref{corollary:large-entropy-remains-uncertified} ensures that
\[
\sum_{t = 1}^{r} \sum_{k = 1}^{r} \Entropy(X_{i}^{j(t), k} \mid M = \mu)
\geq
\Entropy(X_{i}^{j(1)}, \dots, X_{i}^{j(r)} \mid M = \mu)
\geq
2 r^{2} \lg q
\]
and since
$\Entropy(X_{i}^{j(t), k} \mid M = \mu) \leq 3 \lg q$
for every
$(t, k) \in [r] \times [r]$,
it follows that there exists a subset
$\Psi \subseteq [r] \times [r]$
such that
(1) $|\Psi| \geq r^{2} / 2$; and
(2) $\Entropy(X_{i}^{j(t), k} \mid M = \mu) \geq \lg q$
for every
$(t, k) \in \Psi$.

Consider some pair $(t, k) \in \Psi$.
The definition of $\Psi$ guarantees that at least $\lg q$ bits of entropy
remain in the identifier
$\Id(e_{k}(\ell_{i}^{j(t)}))$
of edge
$e_{k}(\ell_{i}^{j(t)})$
after $e^{*}$ is exposed to \Alg{} given that $M = \mu$.
Thus, \Alg{} must have at least $q$ different candidates for
$\Id(e_{k}(\ell_{i}^{j(t)}))$.
The design of the identifier assignment function $\Id(\cdot)$ guarantees that
all but one of these candidate identifiers are actually assigned to dummy
edges and that the candidate identifiers of edge 
$e_{k}(\ell_{i}^{j(t)})$
and the candidate identifiers of edge
$e_{k'}(\ell_{i}^{j(t')})$
are disjoint for every
$(t, k), (t', k') \in \Psi$,
$(t, k) \neq (t', k')$.
Therefore, every edge
$e_{k}(\ell_{i}^{j(t)})$
with
$(t, k) \in \Psi$
that is guaranteed to belong to the edge $(1 - \epsilon)$-cover $F$ output by
\Alg{} contributes at least $q$ distinct edges to $|F|$. 

On the other hand, \Lem{}~\ref{lemma:edge-cardinalities} ensures that the
points in
$e_{k}(\ell_{i}^{j(t)})$
can be covered by at most
$2 q / r \ll q$
edges belonging to
$E'_{-i} = E'_{1} \cup \cdots \cup E'_{i - 1} \cup E'_{i + 1} \cup \cdots \cup
E'_{q + 1}$,
that is, edges corresponding to lines of angles other than $A_{i}$.
Hence, for the sake of the analysis, we may assume hereafter that \Alg{}
covers the points in
$e_{k}(\ell_{i}^{j(t)})$
by edges belonging to $E'_{-i}$
for every
$(t, k) \in \Psi$.

\paragraph{Coverage from another angle.}
Let
$N = \bigcup_{(t, k) \in \Psi} e_{k}(\ell_{i}^{j(t)})$
be the set of points contained in the edges corresponding to the index pairs
in $\Psi$.
Since
$|\Psi| \geq r^{2} / 2$
and since \Lem{}~\ref{lemma:edge-cardinalities} guarantees that
$|e_{k}(\ell_{i}^{j(t)})| \geq (5 / 6) q / r$
for every
$(t, k) \in \Psi$,
it follows that
$|N| \geq 5 q r / 12$.

Recall that the edge $(1 - \epsilon)$-cover $F$ may leave at most
$\epsilon q^{2}$
uncovered points.
The choice of
$r = \lceil 3 \epsilon q \rceil$
ensures that
$\epsilon q^{2} \leq q r / 3$,
thus
at most $q r / 3$ points are not covered by $F$.
The key observation now is that even if all these uncovered points
belong to $N$, then $F$ should still cover at least
$5 q r / 12 - qr / 3 = q r / 12$
points in $N$;
let $N' \subseteq N$ be the subset consisting of these (at least) $q r / 12$
covered points.

We argue that in order to cover the points in $N'$ with edges belonging to
$E_{-i}$, one needs
$\Omega \left( \epsilon q^{2} \frac{\log\log q}{\log q} \right) =
\Omega \left( q r \frac{\log\log q}{\log q} \right)$
distinct edges w.h.p.
The proof of \Lem{}~\ref{lemma:lower-bound-exact-parameters-uncertified} is
completed by union bound since the events described in
\Cor{}~\ref{corollary:large-entropy-remains-uncertified} and
\Lem{}~\ref{lemma:edge-cardinalities} (i.e., the events on which our analysis
is conditioned) hold w.p.\ $\geq 1 / 8$ and w.h.p., respectively.
To that end, consider some line $\ell \in L - A_{i}$, namely, a line from an
angle other than $A_{i}$.
The properties of the affine plane $\mathcal{A}$ ensure that the intersection
$I(\ell) = \ell \cap (\ell_{i}^{j(1)} \cup \cdots \cup \ell_{i}^{j(r)})$
contains exactly $|I(\ell)| = r$ points.
The assignment of these $r$ points to the edges
$e_{1}(\ell), \dots, e_{r}(\ell)$
is determined by the random partition of $\ell$ into
$e_{1}(\ell) \cup \cdots \cup e_{r}(\ell) = \ell$
and it can be viewed as a balls-into-bins process with $r$ balls and $r$ bins.
By a known result on balls-into-bins processes (see, e.g.,
\cite{MitzenmacherU2005}), we conclude that w.h.p.,
$\max_{k \in [r]} |e_{k}(\ell) \cap I(\ell)| = O \left( \frac{\log r}{\log\log
r} \right)$
and by union bound, this holds for all lines $\ell \in L - A_{i}$ w.h.p.;
in particular, every edge in $E'_{-i}$ covers
$O \left( \frac{\log r}{\log\log r} \right)$
points in $N'$ .
The argument follows since
$|N'| = \Omega (q r)$.

This concludes the proof of
\Lem{}~\ref{lemma:lower-bound-exact-parameters-uncertified}.
\Thm{}~\ref{theorem:lower-bound-uncertified} then follows by combining
(\ref{equation:opt-cardinality-uncertified}) and
\Lem{}~\ref{lemma:lower-bound-exact-parameters-uncertified}.

\clearpage

\renewcommand{\thepage}{}
\appendix

\renewcommand{\theequation}{A-\arabic{equation}}
\setcounter{equation}{0}

\begin{center}
\textbf{\large{APPENDIX}}
\end{center}

\section{Proving \Lem{}~\ref{lemma:expected-entropy}}
\label{appendix:proof-lemma-expected-entropy}
Assume first that $n / k = d$ for some integer $d \geq 1$.
Let $\mathcal{S}(n, k)$ be the collection of all ${n}\choose{k}$ subsets $S
\subseteq [n]$ of cardinality $|S| = k$.
By Baranyai's Theorem (see, e.g., \cite{vanLintW2001}), there exists a
partition $\mathcal{P}$ of $\mathcal{S}(n, k)$ into ${{n}\choose{k}} / d$
pairwise disjoint clusters such that every cluster $C$ of $\mathcal{P}$
consists of $d$ subsets
$S \in \mathcal{S}(n, k)$
whose union satisfies
$\bigcup_{S \in C} S = [n]$.
Note that by definition, the subsets in $C$ must be pairwise disjoint.

\sloppy
Given some subset $S = \{ j_{1}, \dots, j_{\ell}\} \subseteq [n]$, let
$X_{S}$ denote the joint random variable
$(X_{j_{1}}, \dots, X_{j_{\ell}})$.
Fix some cluster $C = \{ S_{1}, \dots, S_{d} \}$ of $\mathcal{P}$.
The chain rule of conditional entropy implies that
\begin{align*}
\Entropy \left( X_{1}, \dots, X_{n} \mid Y \right)
~ = ~ &
\Entropy \left( X_{S_{1}} \mid Y \right)
+
\Entropy \left( X_{S_{2}} \mid X_{S_{1}} \mid Y \right)
+ \cdots +
\Entropy \left( X_{S_{d}} \mid X_{S_{1} \cup \cdots \cup S_{d - 1}} \mid Y \right) \\
~ \leq ~ &
\Entropy \left( X_{S_{1}} \mid Y \right)
+
\Entropy \left( X_{S_{2}} \mid Y \right)
+ \cdots +
\Entropy \left( X_{S_{d}} \mid Y \right) \, .
\end{align*}
Denoting the clusters of $\mathcal{P}$ by
$C^{1}, \dots, C^{{{n}\choose{k}} / d}$
and letting
$C^{i} = \{ S_{1}^{i}, \dots, S_{d}^{i} \}$
for $i = 1, \dots, {{n}\choose{k}} / d$,
we can sum over all clusters of $\mathcal{P}$ to conclude that
\begin{equation} \label{equation:entropy-many-terms}
\frac{{{n}\choose{k}}}{d} \Entropy \left( X_{1}, \dots, X_{n} \mid Y \right)
~ \leq ~
\sum_{i = 1}^{{{n}\choose{k}} / d}
\sum_{j = 1}^{d}
\Entropy \left( X_{S_{j}^{i}} \mid Y \right) \, .
\end{equation}
The assertion follows since the right hand side of
(\ref{equation:entropy-many-terms}) has ${n}\choose{k}$ terms, each identified
with a unique subset
$S \in \mathcal{S}(n, k)$,
hence if we pick one term u.a.r., then its expected value is at least
$\Entropy(X_{1}, \dots, X_{n} \mid Y) / d$.
\par\fussy

Now, assume that $n = k \cdot d - r$ for some integers $d \geq 1$ and $0 < r
< k$ and let $n' = k \cdot d$.
Let $X_{n + 1}, \dots, X_{n'}$ be $r$ \emph{dummy} random variables with $0$
entropy.
We have all ready showed that if subset $S \subseteq [n']$ is chosen u.a.r.\
from $\mathcal{S}(n', k)$, then
\[
d \cdot \Expectation_{S} \left[ \Entropy \left( X_{S} \mid Y \right) \right]
~ \geq ~
\Entropy \left( X_{1}, \dots, X_{n'} \mid Y \right)
~ = ~
\Entropy \left( X_{1}, \dots, X_{n} \mid Y \right) \, .
\]
Since
$\Entropy \left( X_{S} \mid Y \right) = \Entropy \left( X_{S \cap [n]} \mid Y
\right)$
for every $S \in \mathcal{S}(n', k)$, it follows that shifting the probability
mass in a uniform manner from subsets $S$ containing dummy variables to
subsets $S$ that do not contain dummy variables cannot decrease the expected
entropy;
in other words, if subset $S \subseteq [n]$ is chosen u.a.r.\ from
$\mathcal{S}(n, k)$ and subset $S' \subseteq [n']$ is chosen u.a.r.\ from
$\mathcal{S}(n', k)$, then
\[
\Expectation_{S} \left[ \Entropy \left( X_{S} \mid Y \right) \right]
~ \geq ~
\Expectation_{S'} \left[ \Entropy \left( X_{S} \mid Y \right) \right] \, .
\]
The assertion follows since $d = \lceil n / k \rceil$.

\clearpage

\renewcommand{\thepage}{}

\bibliographystyle{abbrv}
\bibliography{references}

\begin{thebibliography}{10}

\bibitem{AG09}
K.~Ahn and S.~Guha.
\newblock Graph sparsification in the semi-streaming model.
\newblock In {\em ICALP}, pages 328--338, 2009.

\bibitem{AlonAABN2009}
N.~Alon, B.~Awerbuch, Y.~Azar, N.~Buchbinder, and J.~Naor.
\newblock The online set cover problem.
\newblock {\em SIAM J. Comput.}, 39(2):361--370, 2009.

\bibitem{AlonEFT2013}
N.~Alon, Y.~Emek, M.~Feldman, and M.~Tennenholtz.
\newblock Adversarial leakage in games.
\newblock {\em SIAM J. Discrete Math.}, 27(1):363--385, 2013.

\bibitem{EmekHR2012}
Y.~Emek, M.~M. Halld{\'o}rsson, and A.~Ros{\'e}n.
\newblock Space-constrained interval selection.
\newblock In {\em ICALP (1)}, pages 302--313, 2012.

\bibitem{EpsteinLMS10}
L.~Epstein, A.~Levin, J.~Mestre, and D.~Segev.
\newblock Improved approximation guarantees for weighted matching in the
  semi-streaming model.
\newblock In {\em STACS}, pages 347--358, 2010.

\bibitem{Feigenbaum2005}
J.~Feigenbaum, S.~Kannan, A.~McGregor, S.~Suri, and J.~Zhang.
\newblock On graph problems in a semi-streaming model.
\newblock {\em Theor. Comput. Sci.}, 348:207--216, 2005.

\bibitem{FKMSZ08}
J.~Feigenbaum, S.~Kannan, A.~McGregor, S.~Suri, and J.~Zhang.
\newblock Graph distances in the data-stream model.
\newblock {\em SIAM J. Comput.}, 38(5):1709--1727, 2008.

\bibitem{FraigniaudHPRR2013}
P.~Fraigniaud, M.~M. Halld{\'o}rsson, B.~Patt-Shamir, D.~Rawitz, and
  A.~Ros{\'e}n.
\newblock Shrinking maxima, decreasing costs: New online packing and covering
  problems.
\newblock In {\em APPROX-RANDOM}, pages 158--172, 2013.

\bibitem{GrandoniGLMSS2013}
F.~Grandoni, A.~Gupta, S.~Leonardi, P.~Miettinen, P.~Sankowski, and M.~Singh.
\newblock Set covering with our eyes closed.
\newblock {\em SIAM J. Comput.}, 42(3):808--830, 2013.

\bibitem{HalldorssonHLS10}
B.~V. Halld{\'o}rsson, M.~M. Halld{\'o}rsson, E.~Losievskaja, and M.~Szegedy.
\newblock Streaming algorithms for independent sets.
\newblock In {\em ICALP}, pages 641--652, 2010.

\bibitem{JiaLNRS2005}
L.~Jia, G.~Lin, G.~Noubir, R.~Rajaraman, and R.~Sundaram.
\newblock Universal approximations for tsp, steiner tree, and set cover.
\newblock In {\em STOC}, pages 386--395, 2005.

\bibitem{Karp1972}
R.~M. Karp.
\newblock {Reducibility Among Combinatorial Problems}.
\newblock In R.~E. Miller and J.~W. Thatcher, editors, {\em Complexity of
  Computer Computations}, pages 85--103. Plenum Press, 1972.

\bibitem{KelnerL2013}
J.~A. Kelner and A.~Levin.
\newblock Spectral sparsification in the semi-streaming setting.
\newblock {\em Theory Comput. Syst.}, 53(2):243--262, 2013.

\bibitem{kmm12}
C.~Konrad, F.~Magniez, and C.~Mathieu.
\newblock Maximum matching in semi-streaming with few passes.
\newblock In {\em APPROX}, pages 231--242, 2012.

\bibitem{KremerNR1999}
I.~Kremer, N.~Nisan, and D.~Ron.
\newblock On randomized one-round communication complexity.
\newblock {\em Computational Complexity}, 8(1):21--49, 1999.

\bibitem{LindnerR2011}
C.~C. Lindner and C.~A. Rodger.
\newblock {\em Design Theory}.
\newblock Discrete Mathematics and its Applications. CRC Press, 2nd edition,
  2011.

\bibitem{McGregor05}
A.~McGregor.
\newblock Finding graph matchings in data streams.
\newblock In {\em APPROX-RANDOM}, pages 170--181, 2005.

\bibitem{MitzenmacherU2005}
M.~Mitzenmacher and E.~Upfal.
\newblock {\em {Probability and Computing: Randomized Algorithms and
  Probabilistic Analysis}}.
\newblock Cambridge University Press, Jan. 2005.

\bibitem{SahaG2009}
B.~Saha and L.~Getoor.
\newblock On maximum coverage in the streaming model {\&} application to
  multi-topic blog-watch.
\newblock In {\em SDM}, pages 697--708, 2009.

\bibitem{vanLintW2001}
J.~H. van Lint and R.~M. Wilson.
\newblock {\em A Course in Combinatorics}.
\newblock Cambridge University Press, 2nd edition, 2001.

\bibitem{Vazirani2001}
V.~V. Vazirani.
\newblock {\em Approximation algorithms}.
\newblock Springer-Verlag New York, Inc., New York, NY, USA, 2001.

\end{thebibliography}

\end{document}